\documentclass[11pt]{article}
\usepackage{fullpage}
\usepackage{hdb_macros}

\renewcommand{\tilde}{\widetilde}

\title{Online Orthogonal Vectors Revisited}
\author{
    Karthik Gajulapalli%
    \thanks{Georgetown University. Email: \texttt{kg816@georgetown.edu}.}
    \and
	Alexander Golovnev%
	\thanks{Georgetown University. Email: \texttt{alexgolovnev@gmail.com}.}
	\and
	Samuel King%
    \thanks{Georgetown University. Email: \texttt{sik29@georgetown.edu}.}
    \and
    Sidhant Saraogi%
    \thanks{Georgetown University. Email: \texttt{ss4456@georgetown.edu}.}
}
\date{}
\begin{document}
\pagenumbering{roman}
\maketitle
\thispagestyle{empty}

\begin{abstract}
We prove new upper and lower bounds for the Online Orthogonal Vectors Problem ($\OnlineOV_{n,d}$). In this problem, a preprocessing algorithm receives $n$ vectors $x_1,\ldots,x_n\in\{0,1\}^d$ and constructs a data structure of size $S$. A query algorithm subsequently receives a query vector $q\in\{0,1\}^d$ and in time $T$ decides whether $q$ is orthogonal to any of the input vectors $x_i$.

We design a new deterministic data structure for $\OnlineOV_{n,d}$. In low dimensions ($d = c \log n$), our data structure matches the performance of the best known randomized algorithm due to Chan [SoCG 2017]. Furthermore, in moderate dimensions ($d=n^{\eps}$), we give the first improvement since Charikar, Indyk and Panigrahy [ICALP 2002]. Along the way, we give the first deterministic refutation of a conjecture on the hardness of \OnlineOV{} posed by Goldstein, Lewenstein and Porat [ISAAC 2017]. This data structure also extends to a number of problems, including Partial Match, Orthogonal Range Search, and DNF Evaluation. We use a novel structure-versus-randomness decomposition to design our algorithm. 

Under the Non-Uniform Strong Exponential Time Hypothesis, we also prove arbitrarily large polynomial space lower bounds for any \OnlineOV{} data structure with sublinear query time even with computationally unbounded preprocessing. These lower bounds extend to several other problems, including Polynomial Evaluation, Partial Match, Orthogonal Range Search, and Approximate Nearest Neighbors. We also prove similar lower bounds for \ThreeSUM{} with preprocessing under the Non-Uniform Hamiltonian Path Conjecture.
\end{abstract}

\newpage

\tableofcontents
\thispagestyle{empty}
\newpage
\pagenumbering{arabic}

\section{Introduction}
In the Orthogonal Vectors problem ($\OV_{n,d}$), the task is to determine whether a set of $n$ input vectors $x_1,\ldots,x_n\in\{0,1\}^d$ contains a pair of orthogonal vectors. \OV{} is a central problem in fine-grained complexity, and its associated hardness assumptions have been used to establish tight bounds for a wide range of problems~\cite{W05,V15,V18,GIKW17,CW19}. The most interesting settings for the dimension parameter $d$ 
are the low-dimensional case ($d=c\log{n}$ where $c > 1$ is a large constant) and the moderate-dimensional case ($d=n^\eps$ for some small constant $\eps>0$). The best known algorithms for \OV{} run in time $n^{2-\Omega(1/\log{c})}$ when $d=c\log{n}$~\cite{ILPS14,C14,AWY14,CW16}. 

Beyond its extensive algorithmic applications, \OV{} also finds applications in complexity lower bounds. The Orthogonal Vectors Conjecture (OVC), which posits that there are no subquadratic-time algorithms for $\OV$, is the basis for the fine-grained hardness of a number of problems in $\cc{P}$ (see, e.g., \cite{V18}). A recent beautiful result~\cite{W24} shows that OVC also implies strong circuit lower bounds.  Alternatively, a subquadratic-time algorithm for $\OV$ in $d=\omega(\log{n})$ dimensions would imply strong circuit lower bounds~\cite{W05,JMV15,ABDN18}, resolving an open question posed in 1977~\cite{V77}. 

In this paper, we study the ``online'' version of $\OV$, known as the Online Orthogonal Vectors problem ($\OnlineOV$). $\OnlineOV_{n,d}$ is a problem to be solved in two phases by a pair of algorithms $(\mathcal{P}, \mathcal{A})$. In the \emph{preprocessing phase}, the preprocessing algorithm $\mathcal{P}$ receives $n$ input vectors $x_1,\ldots,x_n\in\{0,1\}^d$ and constructs a data structure $\sigma\in\{0,1\}^S$ of size $S$. During the \emph{query} phase, the query algorithm $\mathcal{A}$, with access to $\sigma$, receives a query vector $q\in\{0,1\}^d$ and in time $T$ decides if $q$ is orthogonal to at least one input vector $x_i$. Such a pair of algorithms is called an $(S,T)$-data structure for $\OnlineOV_{n,d}$, and its preprocessing time $T_{\mathcal{P}}$ is the running time of the algorithm~$\mathcal{P}$. 

The rigorous study of \OnlineOV{} dates back to the 1970s~\cite{K73,R74}, and the problem has since found applications in diverse domains such as information retrieval, databases, and networking (see \cite[Section 6.5]{CIP02,HKP11,AK20,K73} and references therein). It~has been studied under different names, including the Subset Query Problem~\cite{CIP02} and Orthogonal Vectors Indexing~\cite{GLP17}, but typically as a variant of the \PM{} problem~\cite{K73,R74,R76,CIP02,CGL04,C17,AJW25}.\footnote{In the $\PM_{n,d}$ problem, the preprocessing algorithm receives $n$ vectors $x_1,\ldots,x_n\in\{0,1\}^d$, and the query algorithm processes a query $y\in\{0,1,*\}^d$, checking if there is an input vector $x_i$ that exactly matches $y$ in the positions where $y$ is not a wildcard ($*$). 
It is easy to efficiently reduce $\OnlineOV_{n,d}$ to $\PM_{n,d}$, and $\PM_{n,d}$ to $\OnlineOV_{n,2d}$.} 
\OnlineOV{} is essentially equivalent to the Orthogonal Range Search, DNF Evaluation, Maximum Inner Product, and variants of the Nearest Neighbors problems~\cite{I98,CIP02,AWY14,C17,CW19}. The most interesting settings of the dimension parameter for $\OnlineOV_{n,d}$ are also $d=c\log{n}$ and $d=n^\eps$, as in the case of \OV{}.

There are two trivial algorithms for $\OnlineOV_{n,d}$. The first  simply stores answers to all possible queries in space $S=2^d$ and answers each query in time $T=d$. The second algorithm stores only input vectors, using space $S=nd$, and answers each query by scanning the entire input in time $T=nd$. 

The best known deterministic algorithms for $\OnlineOV_{n,d}$ are those by~\cite{CIP02}, which run in (i) space $S=n\cdot2^{O(d\log^2(d)\sqrt{t/\log{n}})}$ and time $T=n/2^t$ for any $t$; and (ii) space $S=nd^t$ and time $T=O(nd/t)$ for $t\leq n$. Notably, when $d=c\log{n}$ for a large enough constant $c>1$, a more efficient \emph{randomized} algorithm~\cite{C17} achieves space $S=n^{1.1}$ and query time $T=n^{1-1/(c\log(c))}$. A similar running time of $T=n^{1-1/(c\log^2(c))}$ is achieved by another very recent randomized algorithm~\cite{AJW25}.

Several conjectures have emerged regarding the hardness of \OnlineOV{} in both the worst and average cases~\cite{R76,BOR99,I01,GKLP17,GLP17,LS22,JPP25}. For example, the strong version of the \emph{``curse of dimensionality'' conjecture} asserts that for $d=\omega(\log{n})$, there is no $(S,T)$-data structure for $\OnlineOV_{n,d}$ that simultaneously satisfies $S=2^{o(d)}$ and $T=o(n)$. This conjecture was refuted for all $d=2^{o(\log(n)^{1/4})}$ by the first algorithm in~\cite{CIP02}.

In the case of logarithmic dimension $d=c\log{n}$, \cite{GLP17} designed deterministic data structures for $\OnlineOV_{n,d}$ with space and time complexities $S=n^{c-0.3}$ and $T=n^{1-\delta}$ in the worst case, and $S=n^{c-0.99}$ and $T=n^{1-\delta}$ in the average case. This led them to propose the following conjecture.
\begin{mainconjecture}[{\cite[Conjecture~2]{GLP17}}]\label{conj:introGLP}
    For every $\eps>0$ there is $c>1$ such that  no $(S,T)$-data structure solves $\OnlineOV_{n,c\log{n}}$ with $S=n^{c-1}$ and $T=n^{1-\eps}$.
\end{mainconjecture}
The results of~\cite{C17,AJW25} refute this conjecture, albeit by randomized algorithms, i.e., only against oblivious adversaries whose query points are independent of the random choices made by the preprocessing algorithm of the data structure. 

There has been significant progress in proving lower bounds for restricted models~\cite{I98,P08,P11,ALRW17,AK20,AJW25}. 
In particular, \cite{P08,P11} shows that any decision tree that solves $\OnlineOV_{n,d}$ must be of size $2^{\Omega(d)}$ or have depth $\Omega(n^{1-2\delta}/d)$. And \cite{AJW25} shows that any List-of-Points data structure that solves $\OnlineOV_{n,c\log n}$ must use at least $n^{\omega(1)}$ space or have query time $\Omega(n^{1-1/\sqrt{c}})$.\footnote{A List-of-Points data structure \cite{ALRW17, AK20} is a ``data-independent'' model where the data structure cannot encode its input but rather can only store lists of data points that intersect with previously-chosen lists.}

For data structures with \emph{efficient preprocessing}, lower bounds for \OnlineOV{} can be derived from well-known assumptions, such as the Strong Exponential Time Hypothesis (SETH).\footnote{SETH asserts that for every $\eps>0$, there exists $k$ such that $k$-SAT on~$n$ variables requires time at least $(2-\eps)^n$.} One can get arbitrarily large polynomial space lower bounds~\cite{GLP17} using Williams' celebrated reduction from $k$-SAT to \OV{}~\cite{W05}. That is, under SETH, for all constants $d, \delta > 0$, there exists $c > 1$ such that there are no data structures with query time $n^{1-\delta}$ and preprocessing time (and thus space) $n^d$ for $\OnlineOV_{n,c\log{n}}$.

While fine-grained complexity assumptions have been successfully leveraged to establish lower bounds for data structures with computationally efficient preprocessing (see, e.g., ~\cite{AV14, HKNS15}), there is still no clear path to prove ``space lower bounds''---lower bounds on the space complexity of data structures with unbounded preprocessing. The framework of conjectures introduced by \cite{GKLP17} aims to address such time-space tradeoffs, but to date, it has not led to connections with other key problems in fine-grained complexity.

In the case of \emph{computationally unbounded preprocessing}, despite much effort, the best known lower bound remains $T\geq \Omega(\log n)$, even when the space of the data structure is linear~\cite{MNSW95,BOR99,BR00,JKKR03,PT06,PTW08,P11}. In fact, this bound matches the best known lower bound for \emph{any} explicit data structure problem in the cell probe model~\cite{S04,L12}. Establishing a stronger lower bound for an explicit data structure problem with unbounded preprocessing remains a major open problem in the field.

The central challenge in the field of static data structure lower bounds is the following.
\begin{mainopenproblem}\label{op:lowerbounds}
For an explicit data structure problem with input size $n$, show that any data structure which uses space $S = cn$ for some constant $c > 1$ must have query time $T \geq n^{\epsilon}$ for some constant $\varepsilon > 0$ even with computationally unbounded preprocessing. 
\end{mainopenproblem}
While proving such a lower bound unconditionally is beyond the reach of current techniques~\cite{DGW19, V19, NR20}, one can hope to prove such lower bounds under well-studied conjectures.

\subsection{Our Results}

\subsubsection{Upper Bounds}

We design new \emph{deterministic} data structures for the Online Orthogonal Vectors problem ($\OnlineOV{}$) which work in arbitrary dimension $d$. We make two interesting contributions. First, the data structure matches the performance of the best known data structures in the case where the dimension $d(n) = c \log n$. %
In contrast to previous state-of-the-art data structures~\cite{C17,AJW25}, our data structure is fully deterministic.
\begin{theorem}[Informal, c.f.~\cref{cor:chan_comparison}]\label{thm:alg_inf}
 Consider any $c > 1$, sufficiently large $n$ and any $\delta > \Omega(\log c/c)$. There is an $(S, T)$-data structure with preprocessing time $T_p$ for $\OnlineOV_{n,c \log n}$ where: 
    \begin{align*}
       T & = \tilde{O}\left(n^{1-\frac{1}{O(c\log c)}}\right)\;,   & \text{and} && S, \;T_p & = \tilde{O}(n^{1+\delta}) \;.
    \end{align*}
\end{theorem}
\cref{thm:alg_inf} gives the first \emph{deterministic} refutation of \cref{conj:introGLP}, while the previous results~\cite{C17,AJW25} refuted the conjecture against oblivious adversaries only. (However, the data structures designed by \cite{C17,AJW25} work for more general problems than \OnlineOV{}.) In fact, \cref{thm:alg_inf} refutes an even weaker conjecture \cite[Conjecture~13]{GKLP17} stating that no $(S,T)$-data structure solves $\OnlineOV_{n, c\log{n}}$ with $S=n^{2-\Omega(1)}$ and $T=n^{1-\eps}$.

Our second main contribution is to provide improved data structures for $\OnlineOV$ in arbitrary dimension $d$. 
\begin{theorem}[Informal, c.f.~\cref{cor:cip_comparison}]\label{thm:infCIP}
    Consider any $n, d$ and $\epsilon < 1/2$, then there exists an $(S, T)$-data structure for $\OnlineOV_{n, d}$ with preprocessing time $T_p$ such that:
    \begin{align*}
        T & \leq n^{1-\eps}d\;, & \; S & = n^{1-\eps}d\,2^{O(\eps d\log(1/\eps))} \;,
        & T_p & \leq n^{1+3\eps}d\,2^{O(\eps d\log(1/\eps))}\;.
    \end{align*}
\end{theorem}
\cref{thm:infCIP} improves substantially on the best data structures in moderate dimensions provided by \cite{CIP02} which have $S \leq n\cdot2^{O(\sqrt{\eps}d \log^2 d)}$ and query time $n^{1-\eps} \cdot d$. This allows us to refute the strong version of the curse of dimensionality (which requires $2^{o(d)}$ space and $o(n)$ query time) for any $d = n^{o(1)}$.\footnote{Note that we can pick $\varepsilon = 2\log d / \log n = o(1)$. Then $T \leq n/d = o(n)$. Furthermore, note that $\varepsilon \cdot \log(1/\varepsilon) = o(1)$ and therefore $S = O(2^{o(d)})$.}

Finally, we can also combine our data structures with known reductions from various online problems to $\OnlineOV$ to obtain new deterministic data structures for these problems (\cref{apps:deterministic_applications}). These problems include Partial Match, DNF-Evaluation, Subset Query, Containment Query, and Approximate Orthogonal Range Searching (See~\cref{section:ds_reductions} for definitions). We also provide a randomized data structure for the regular version of Orthogonal Range Searching in this manner (\cref{apps:ors}). All of our data structures can be easily modified to work for \emph{reporting} versions of the problem, i.e., returning all the orthogonal vectors in the input as opposed to identifying if one exists, with minimal overhead.

\paragraph{Techniques.} Our main insight is a simple but powerful observation about the structure of random input vectors. We are able to leverage this with a \emph{structure versus randomness} approach to design our data structure. While our data structure is simple, the analysis of its performance is somewhat tricky. 

Our algorithm also has the benefit of being \emph{combinatorial}, i.e., it does not resort to algorithms for fast matrix multiplication. In fact, our algorithm matches the performance of the best known combinatorial algorithm for \OV{} even in the ``offline'' setting~\cite{C14}. %

In the process of designing a worst-case data structure, we also design new data structures for $\OnlineOV{}$ when the $n$ input vectors $X \subseteq \{0, 1\}^d$ are sampled such that each entry in each vector is set to $0$ independently with probability $p$ (let's call this distribution $\mathcal{B}_p$). Here, we highlight the interesting case where $d = c \log n$ and $p = 1/2$ but our data structure works for arbitrary values of $d$ and~$p$. 

\begin{theorem}[Informal, c.f. \cref{cor:chan_comparison_avg}]
    Let $n \geq 1$, $c > 1, \delta > \Omega(\log c/c)$. 
    Then, there is an average-case  $(S, T)$-data structure for $\OnlineOV_{n,c \log n}$ over $\mathcal{B}_{1/2}$ with preprocessing time $T_p$ where:
    \begin{align*}
            T &\leq \tilde{O}\left(n^{1-\frac{1}{O(\log c)}}\right)\;, && \text{and} &
            S, \; T_p &\leq \tilde{O}(n^{1+\delta})\;.
    \end{align*}
\end{theorem}

The simplicity of the algorithm combined with its substantially better guarantees than our worst-case algorithm (where $T = \tilde{O}(n^{1-1/O(c \log c)})$) raises the possibility of designing better data structures using an improved structure versus randomness observation. 
In particular, we can use our data structure to design an average-case algorithm for ``offline'' $\OV_{n, c \log n}$ in time $O(n^{2- \Omega(1/\log c)})$ (\cref{apps:avg-ov}), which comes close to the best known algorithms for the offline average-case version of the problem~\cite{KW19,AAZ25}.

\subsubsection{Lower Bounds}

We make progress towards \cref{op:lowerbounds} by proving strong lower bounds for data structures with computationally unbounded preprocessing, under two non-uniform conjectures.
Our first set of results proves arbitrary polynomial lower bounds for a host of popular online data structure problems under the Non-Uniform SETH (NUSETH) conjecture.\footnote{A non-uniform version of SETH was first considered by \cite{CGIMPS16} where they considered the Non-Uniform Nondeterministic Strong Exponential Time Hypothesis.}

\begin{conjecture}[NUSETH]\label{conj:nuseth_informal}
    For every $\eps>0$, there exists $k$ such that no non-uniform algorithm solves $k$-SAT in time $2^{(1-\eps)n}$ using $2^{(1-\eps)n}$ bits of advice.\footnote{See \cref{def:non-uniform-alg} for a formal treatment of a non-uniform algorithm.}
\end{conjecture}

Our first result shows that the hardness of $\OnlineOV$ in the high dimensional regime holds even for instances with constant input sparsity. Specifically, under NUSETH any batch data structure\footnote{For a formal definition of a batch data structure, see \cref{def:batch_ds}.} for Sparse-$\OnlineOV_{n, n^{\delta}}$ with computationally unbounded preprocessing cannot achieve both fixed polynomial space and sub-linear amortized query time.  Note that since a batch data structure is a stronger model of computation, the lower bound on batch data structures also holds for standard $(S, T)$-data structures.

\begin{theorem}[Informal, c.f. \cref{thm:batchDNF}]\label{thm:seth_lb_informal}
    Under \rm{NUSETH}, for all constants $c,\alpha>0$ and $\delta,\gamma\in(0,1)$, there exists constant~$w$ such that no batch data structure with computationally unbounded preprocessing can answer a set of $n^\alpha$ queries of $w$-Sparse-$\OnlineOV_{n,d}$ for $d=n^\delta$ in space $n^c$ and time $n^{1-\gamma}$ per query.
\end{theorem}

We remark that our lower bound framework is quite robust, as it also recovers the previous lower bound of Abboud and Williams ~\cite{AV14} in the low dimensional regime  $(d = \poly (\log n))$. In fact, \cref{thm:seth_lb_informal} can be viewed as a generalization of the lower bound in \cite{AV21} that additionally takes sparsity into account (see \cref{rem:ov_low_dim}).

\begin{corollary}[Informal, c.f. \cref{cor:av_21} and also \cite{AV21}]
    Under {NUSETH}, for all constants $c,\alpha>0$, $\gamma\in(0,1)$, there exists a constant $\rho>0$ such that no batch data structure with computationally unbounded preprocessing can answer a set of $n^\alpha$ queries of $\OnlineOV_{n,d}$ for $d=O(\log (n)^\rho)$ in space $n^c$ and amortized time $n^{1-\gamma}$ per query.
\end{corollary}

As a consequence of strong connections between \OnlineOV{} and other well-studied problems including \PM{}, \SubsetQuery{}, \ORS, $(1+\eps)$-\ANN{}, we are able to immediately extend our lower bounds to these problems as well.\footnote{We refer the reader to \cref{section:ds_reductions} for formal definitions of these problems.}

\begin{corollary}\label{cor:intro_strong_nuseth_lb}
Under {NUSETH}, for all constant $c,\alpha>0$, $\gamma\in(0,1)$, there exists a constant $\rho>0$ such that no batch data structure with computationally unbounded preprocessing can answer a set of $n^\alpha$ queries in space $n^c$ and amortized time $n^{1-\gamma}$ per query for any of the following problems with $d=O(\log(n)^\rho)$:
\begin{align*}& \PM_{n, d}, \; \SubsetQuery_{n, d}, \; \ContainmentQuery_{n, d}, \\
    & \ORS_{n, d, 1}, \; \DNFEval_{n,d}, \; \OnlineIP_{n, d, 0},\\
     & (1+\eps)\text{-}\ANN_{n, d, p, \infty}\;.
    \end{align*}
\end{corollary}

Several non-uniform versions of SETH and related conjectures have been studied in previous works (see, e.g., \cite{GHLORS12,CEF12,CGIMPS16,ABGS21,AV21}). While it may appear that  {NUSETH} is a strong assumption, we remark that our current techniques cannot even refute the (potentially) much stronger conjecture NUNSETH which posits that there is no $(2-\eps)^n$-time non-uniform \emph{non-deterministic} algorithm that decides $k$-UNSAT for every~$k$.

We now introduce a new conjecture against non-uniform algorithms that we call the Non-Uniform HamPath Conjecture, which loosely states that any non-uniform algorithm solving Hamiltonian Path on $n$ vertices requires time at least $2^n$.\footnote{We refer the reader to \cref{sec:fgc} for a discussion on this conjecture.}

\begin{conjecture} [Non-Uniform HamPath Conjecture, Informal]\label{conj:nuhampath_informal}
    For every constant $\eps\in (0,1)$, no non-uniform algorithm solves \HamPath{} on graphs with $n$ vertices in time $2^{(1-\eps)n}$ using $2^{(1-\eps)n}$ bits of advice.
\end{conjecture}

We use \cref{conj:nuhampath_informal} to prove a lower bound for the \ThreeSUM{} with preprocessing problem. In this problem, the preprocessing algorithm receives three sets: $\mathcal{X}_1, \mathcal{X}_2 , \mathcal{X}_3$ each consisting of $N$ integers, and preprocesses them into a data structure $\sigma\in\{0,1\}^S$ of size~$S$. Then, the query algorithm receives a query $\mathcal{X}_1' \subseteq \mathcal{X}_1, \mathcal{X}_2' \subseteq \mathcal{X}_2, \mathcal{X}_3' \subseteq \mathcal{X}_3$, and decides in time~$T$ if there exist $x_1 \in \mathcal{X}_1', x_2 \in \mathcal{X}_2', x_3 \in \mathcal{X}_3'$ such that $x_1 + x_2 = x_3$.

A line of research~\cite{BW09,CL15,CVX23,KPS25} culminated in an $(S, T)$-data structure for \ThreeSUM{} with preprocessing with $S=T=\widetilde{O}(N^{3/2})$. While this data structure is already very efficient, it is an open question if it can be further improved. We do not have any unconditional lower bounds better than the trivial ones requiring linear time or linear space.
Under the Non-Uniform \HamPath{} Conjecture, we now show that any data structure for \ThreeSUM{} with computationally unbounded preprocessing must require either super-linear space or super-linear query time.
 
\begin{theorem}[Informal, c.f. \cref{thm:hampath}]
    Under \cref{conj:nuhampath_informal}, there is no $(N^{1.087}, N^{1.087})$-data structure for \ThreeSUM{} with computationally unbounded preprocessing.
\end{theorem}

In fact our result generalizes  to \kSUM{} (\cref{thm:hampath}) for any $k \geq 3$ and to the offline version of \kSUM{}.

\subsection{Proof Overview}

\subsubsection{Algorithms}
Our algorithms are tailored for the $\OnlineOV$ problem. For this section, we focus on the case where the input $X \subseteq \{0, 1\}^{d}$ contains $|X| = n$ vectors of dimension $d := c \log n$.  First, we describe a simple randomized algorithm which is efficient on average when the input vectors are picked independently and uniformly at random. 

\paragraph{Average-case algorithm.} We begin with the simple observation that, when $X$ is chosen at random, it is unlikely that there exists a large subset $X' \subseteq X$ and a large set of coordinates $S \subseteq [c\log n]$ such that every vector $x \in X'$ satisfies $x|_S = \mathbf{0}$.
Let $Y_S := \{x \in X \colon\; x|_S = \mathbf{0}\}$. To quantify the observation, we fix some appropriate $\eps > 0$ such that, with high probability, we expect that $|Y_S| < n^{1-\eps}$ for an arbitrary $S$ of certain size~$t$. Note that $\mathbb{E}_{X}[|Y_S|] = n \cdot \frac{1}{2^{|S|}}$. Using the Chernoff and union bounds, $|Y_S| < n^{1-\eps}$ for all $|S| = t\approx \eps \log n$ with high probability.

Armed with this observation, our average-case algorithm is simple to state. In the discussion below, we will focus on the space complexity of the data structure and the query time of the online phase. In \cref{sec:algs}, we show that the preprocessing algorithm is also efficient. In the preprocessing phase, $\texttt{AverageOVPre}$ uses $\binom{d}{\leq t}$ space to store the answers to all possible sparse queries of Hamming weight $\leq t$. Furthermore, it stores the lists of relevant inputs $Y_S$ for all $S \subseteq [c \log n]$ with $|S| = t$. This takes space at most $\binom{c \log n}{t} \cdot n^{1-\eps}$ with high probability. In the online phase, $\texttt{AverageOVOnl}$, if the query~$q$ has Hamming weight $\|q\|_0 < t$, it takes constant time to check the answer among the stored answers. Otherwise, the algorithm picks the first set $S$ of $t$ $1$'s of $q$. These coordinates must be $0$ in any orthogonal input vector and therefore belong to $Y_S$. Now, in time $|Y_S| < n^{1-\eps}$ it is easy to check if one of the vectors in $Y_S$ is orthogonal to~$q$. 

To summarize, the space required by $\texttt{AverageOVPre}$ is at most $\binom{c \log n}{\eps \log n} \cdot n^{1-\eps} \leq n^{1+\eps\log(ec/\eps)}$, and the query time of $\texttt{AverageOVOnl}$ is at most $\Ot(n^{1-\eps})$. %

\paragraph{Pseudorandomness.} For an arbitrary set of input vectors $X$, there is no guarantee that $X$ behaves like a random input. For example, every vector in $X$ could have its first $\frac{c\log n}{2}$ coordinates be $0$. Then one of our stored lists has size $|Y_{[c\log n/2]}| = n$, leading to our average-case algorithm having poor query time. However, violating the assumption that $Y_S$ is small for a particular $S$ leads to finding some structure in our problem. For example, in the above case, we could reduce the dimension of the problem by half. We formalize this intuition by defining a simple notion of pseudorandomness. For any $\eps > 0$, we say that $X$ is $(n^{1-\eps}, t)$-pseudorandom if for every $S \subseteq [c \log n]$ with $|S| = t$, we have that $|Y_S| \leq n^{1-\eps}$. While we cannot expect $X$ itself to be pseudorandom, we find a decomposition of $X = X' \sqcup X_1 \sqcup \dots \sqcup X_{n^{\eps}}$ such that: 
\begin{itemize}
    \item For each $X_i$, we have $|X_i| = n^{1-\eps}$. Additionally, we find a corresponding $S_i  \subseteq [c \log n]$, $|S_i|=t$, such that each $x \in X_i$ satisfies $x|_{S_i} = \mathbf{0}$. For the purpose of finding orthogonal vectors in $X_i$, we can ignore the coordinates in $S_i$, i.e. $X_i$ has some structure. 
    \item The remaining part $X'$ of $X$ is $(n^{1-\eps}, t)$-pseudorandom. 
\end{itemize}

\paragraph{Worst-case Algorithm.} We now describe our worst-case algorithm. First, using a greedy procedure, we find an appropriate decomposition of $X = X' \sqcup X_1 \sqcup \dots \sqcup X_{n^{\eps}}$ as above. As in the average-case algorithm, in the preprocessing phase, $\texttt{OVPre}$ (computes and) stores: 
\begin{enumerate}
    \item the answers to all possible sparse queries $q$ with $\|q\|_0 < t$ in space $\binom{c \log n}{< t}$, and
    \item for each $S \subseteq [c \log n]$ with $|S| = t$, a list of candidate orthogonal vectors $Y_S$ from $X'$ with $|Y_S| \leq n^{1-\eps}$ in total space $\binom{c \log n}{t} \cdot n^{1-\eps}$.
\end{enumerate}
From this information, as in the average-case algorithm, we can answer all sparse queries and check any queries against the candidate vectors in the pseudorandom part $X'$.
However, this leaves us with all the remaining candidate vectors in the structured parts $X_1, \dots, X_{n^{\eps}}$.
For each set of vectors $X_i$, we can drop the coordinates in $S_i$ as they are always $0$. 
In other words, each $X_i$ defines a subproblem of $\OnlineOV$ with $n^{1-\eps}$ input vectors in $c \log n - t$ dimensions.
Therefore, $\texttt{OVPre}$ recursively creates a data structure for each of the $n^\eps$ sets $X_i$. 
This recursion continues until (A)~the subproblem contains just one vector, in which case the algorithm just stores that vector, or (B)~the dimension gets small enough, in which case the algorithm stores the answers to all possible queries.
We can then use induction to argue that each recursive data structure requires space $\approx \binom{c\log n}{t}\cdot n^{1-2\eps}$ for $\binom{c\log n}{t} \cdot n^{1-\eps}$ space total across all recursively-constructed data structures, about matching the amount of space needed for the pseudorandom part.

During query time, if $\|q\|_0 < t$, $\texttt{OVOnl}$ simply looks up the answer.
Otherwise, it finds the first $t$ non-zero coordinates of $q$ and checks the corresponding list $Y_S$ of candidates from $X'$, which takes time at most $n^{1-\eps}$. 
Finally, it determines whether there are any orthogonal vectors in each $X_i$ using the recursively-constructed data structures. 
By induction, we argue that each of these subproblems takes time at most $\approx n^{1-2\eps}$ for $n^{1-\eps}$ time total across all of the structured parts. 

For technical reasons, we need to vary the parameters $\eps$ and $t$ at different levels of recursion.
To do this, we introduce a parameter $i$ to track how many levels of recursion you are above the base case ($i=1$ is at the base case), and we set $\eps = 1/i$ and $t = c \log n / i$.
Thus, setting the top level of recursion to be at $i \approx c \log c$, we get query time $T \leq n^{1-\Omega\left(\delta/c\log c\right)}$ and space $S \leq O(n^{1+\delta})$ where $\delta > 0$ is an arbitrarily small constant.
Moreover, $\texttt{OVPre}$ is time efficient, with preprocessing time $T_p \leq O(n^{1+\delta})$.

\subsubsection{Space Lower Bounds}\label{sec:space}
\paragraph{Space Lower Bounds from NUSETH.}

The general framework for proving data structure space lower bounds (under SETH) involves reducing an instance of $k$-SAT to an instance of some online problem. For example, this is the framework used by \cite{GLP17} who apply the classical reduction from SETH to $\OV$  \cite{W05} to prove that there is no $\alpha, \delta > 0$ such that for all $c \geq 1$, there is a $(N^{\alpha}, N^{1-\delta})$ data structure for $\OnlineOV_{N, c \log N}$ with preprocessing time $T_{\P} \leq N^{\alpha}$.

We wish to prove space lower bounds on $(S, T)$-data structures with \emph{computationally unbounded preprocessing}. However, since our preprocessing algorithm is not efficient, the above framework no longer works. It is unclear how to leverage SETH to prove such a lower bound. The first attempt to try and fix this is to assume a stronger non-uniform hypothesis like NUSETH. We could encode the space used by the data structure as advice for a non-uniform algorithm that solves $k$-SAT efficiently and thus breaks NUSETH. However, this assumption alone is still not enough to get a lower bound for $\OnlineOV_{n, c \log n}$. This is because in the known reductions from $k$-SAT on $n$ variables to $\OnlineOV_{n, c \log n}$, every $k$-CNF formula $\varphi$ gets mapped to a different instance of $\OnlineOV_{n, c \log n}$. As a result there is no single non-uniform advice string because the preprocessing algorithm creates a new data structure for each $k$-CNF formula. 

Abboud and Williams \cite{AV21} give one approach to address this challenge by constructing a universal table that encodes the satisfiability information for all possible $k$-CNF formulas over $n$ variables. Specifically, consider a table where each row corresponds to a Boolean assignment to the $n$ variables, and each column corresponds to one of the $\binom{n}{k} \cdot 2^k$ possible $k$-CNF clauses. The columns are ordered according to the lexicographic order of their corresponding clauses. For each row (i.e. assignment), we place a $0$ in columns corresponding to clauses satisfied by the assignment, and a $1$ otherwise.\footnote{We thank Virginia Vassilevska Williams and an anonymous SODA reviewer for pointing out~\cite{AV21} to us.}

We use this table as the input instance $\mathcal{D}$ to the $\OnlineOV$ problem, consisting of $N = 2^n$ Boolean vectors, each of dimension $d = \binom{n}{k} \cdot 2^k$. Given a $k$-CNF formula $\varphi$, we define its indicator vector $I_{\varphi} \in \bit^d$ such that the $i$-th bit of $I_{\varphi}$ is set to $1$ if and only if $\varphi$ contains the $i$-th clause in lexicographic order.
It follows that $\varphi$ is satisfiable if and only if $I_{\varphi}$ is orthogonal to some vector in $\mathcal{D}$. Moreover, the input instance $\mathcal{D}$ is fixed and independent of the particular $k$-CNF formula being queried.

The key idea now is that $\mathcal{D}$ defines a single instance of \OnlineOV{} which can be used to test satisfiability of \emph{every} $k$-SAT formula on $n$-variables. If there was an efficient data structure for fast \OnlineOV{}, then we could transform our data structure into a fast non-uniform algorithm for $k$-SAT. Suppose that for some $\delta > 0$, there is an $(N^{1.99}, N^{1-\delta})$-data structure with unbounded preprocessing for \OnlineOV{}. We apply a standard branching argument based on the downward reducibility of $k$-SAT, which lets us take a $k$-SAT formula $\varphi$ on $n$ variables and reduce it to $2^{n/2}$ formulas each on $n/2$ variables. Alternatively, one can view this as a reduction to \OnlineOV{}, with $N = 2^{n/2}$ input vectors and $K = 2^{n/2}$ query points. We can compute the $K$ query points in time $O^*(2^{n/2})$. Furthermore, all the queries can be solved in time $K \cdot N^{1-\delta} = N^{2-\delta}$. This gives us a non-uniform algorithm that uses $N^{1.99} = 2^{0.995n}$ bits of advice and runs in time $N^{2 - \delta} \ll 2^n$, which refutes NUSETH. Moreover one can adjust the branching factor to prove arbitrarily large polynomial lower bounds on space.

While this already implies strong data structure lower bounds, we can further strengthen the result by combining the above reduction with the polynomial formulation framework introduced in~\cite{BGKMS23}. This allows us to establish arbitrarily large polynomial lower bounds even for instances of \OnlineOV{} with constant input sparsity, albeit at the cost of a sub-linear blowup in the dimension of the input vectors.

Applying the framework of polynomial formulations to $k$-SAT on $n$ variables, one can extract an explicit construction of a special instance  $\mathcal{D}_{N,w}^{\delta}$ of \OnlineOV{}  with the following properties (see \cref{def:hardest_OnlinneOV}):

\begin{itemize}
    \item [1.] $\mathcal{D}_{N, w}^{\delta}$ is an instance of $\OnlineOV_{N, d}$ with $N$ input vectors where $N = 2^{n}$, and dimension $d = N^{\delta}$ for a constant $\delta \in (0, 1]$. Each input vector is $w$-sparse with $w = \binom{k/\delta}{k} = O(1)$, and one can construct $\mathcal{D}_{N, w}$ in time $\Ot(N)$.

    \item [2.] There exists a $2^{(1-\gamma)n}$ algorithm that takes as input a $k$-SAT formula $\varphi$ on $n$ variables and outputs a query vector $q_{\varphi} \in \bit^{N^{\delta}}$ such that there exists a $v \in \mathcal{D}_{N,w}$ such that $v \perp q_{\varphi}$ if and only if $\varphi$ is satisfiable. The running time of this algorithm is $\Ot(N^\delta)$.
\end{itemize}

We refer to our specially constructed $\mathcal{D}_{N,w}$ as the hardest \OnlineOV{} instance. The remainder of the argument then follows directly from the previously described reduction.

\paragraph{Space Lower Bounds from Non-Uniform HamPath.} 
In this section we will show how to use a different non-uniform conjecture (\cref{conj:nuhampath_informal}) to prove a lower bound for an $(S,T)$-data structure for \ThreeSUM{} with preprocessing (\cref{def:3sum_pre}). For ease of exposition we will work with the equivalent version of \ThreeSUM{} where given 3 sets $X_1$, $X_2$ and $X_3$, and a fixed number $t$, the task is to check if there exist $x_1 \in X_1$, $x_2 \in X_2$ and $x_3 \in X_3$ such that $x_1 + x_2  + x_3 = t$.

We begin by showing a uniform reduction from Hamiltonian Path to the offline version of \ThreeSUM{}. Let $G$ be a graph on $n$ vertices.\footnote{Without loss of generality, one can assume that $n$ is a multiple of~$3$.} We use the fact that any Hamiltonian Path in $G$ can be partitioned into $3$ disjoint simple paths of length $n/3$.
Let $\mathcal{S}_{n/3}$ be the collection of all sets of size $n/3$. Now for a set $S\in\mathcal{S}_{n/3}$, and vertices $u,v\in[n]$, we can define the following variable,
    \[
    D[S, u, v] = 
\begin{cases}
    1 & \text{if $u\in S$ and $\exists$ a simple path from $u$ to a neighbor of~$v$ using all vertices in $S$,} \\
    0  & \text{otherwise.}
\end{cases}
\]

The key observation is that $G$ has a Hamiltonian Path if and only if there exists a collection of disjoint sets $S_1$, $S_2$ and $S_3 \in \mathcal{S}_{n/3}$, and a set $(a,b,c,d)$ of nodes in $G$ such that :
    \begin{equation} \label{eq:const_1}
        D[S_1, a, b] \wedge D[S_2, b, c] \wedge D[S_3, c, d] = 1\;,
    \end{equation}
    and
  \begin{equation} \label{eq:const_2}
      \Int{(S_1)} + \Int{(S_2)} + \Int{(S_3)} = 2^n -1 \;.\footnote{$\Int{(S)}$ is the integer value in $[0, 2^n -1]$ whose binary representation matches the $n$-bit indicator of the set $S$.}
  \end{equation}

We now reduce Hamiltonian Path to \ThreeSUM{} with preprocessing as follows.

\begin{description}
    \item[\textbf {Preprocessing Phase.}]
    In the preprocessing step we define the following sets that are independent of the graph $G$ and just depend on $n$:
    \[\mathcal{X}_1 = \mathcal{X}_2  = \mathcal{X}_3 = \{\Int{(S)} :  S \in \mathcal{S}_{n/3} \}\;.\]
    \item [\textbf{Query Phase.}]
    Now for every quadruple of vertices $(a, b, c, d)$ we construct the following queries to our \ThreeSUM{} instance with $t = 2^n - 1$:
        \begin{align*}& \mathcal{X}'_1 = \{\Int{(S_i)} \mid S_i \in \mathcal{S}_{n/3} \text{ and } D[S_i, a, b] = 1\}, \\
    & \mathcal{X}'_2 = \{\Int{(S_i)} \mid S_i \in \mathcal{S}_{n/3} \text{ and } D[S_i, b, c] = 1\}, \\
     & \mathcal{X}'_3 = \{\Int{(S_i)} \mid S_i \in \mathcal{S}_{n/3} \text{ and } D[S_i, c, d] = 1 \}\;.
    \end{align*}
By construction we have that $\mathcal{X}'_1 \subseteq \mathcal{X}_1$, $\mathcal{X}'_2 \subseteq \mathcal{X}_2$, and $\mathcal{X}'_3 \subseteq \mathcal{X}_3$ and hence denotes a valid \ThreeSUM{} with preprocessing query. Our \ThreeSUM{} query returns $1$ if only if there exists $x_1 \in \mathcal{X}'_1$, $x_2 \in \mathcal{X}'_2$ and $x_3 \in \mathcal{X}'_3$ such that $x_1 + x_2 + x_3 = 2^n -1$. 
    
\end{description}

Our \ThreeSUM{} instance now encodes the constraints \cref{eq:const_1,eq:const_2}. We report that the graph has a Hamiltonian Path if and only if the \ThreeSUM{} returns $1$ for any query. 

We now analyze our reduction. We construct a \ThreeSUM{} with preprocessing instance of size $N = |\mathcal{S}_{n/3}| \leq \binom{n}{n/3} \leq 2^{nH(1/3)}$. We construct $n^4$ queries to our \ThreeSUM{} instance, and the cost of constructing each query is $\Ot(N)$. Now we show that our reduction implies that there is no $(N^{1.08}, N^{1.08})$ data structure for \ThreeSUM{} with preprocessing, because we can convert such a data structure into a fast non-uniform algorithm $\mathcal{A}$ for Hamiltonian Path. Assume towards a contradiction that such a data structure exists, then our algorithm $\mathcal{A}$ encodes the $S$ bits from the data structure as advice, with $S \leq 2^{1.08(nH(1/3))} \ll 1.99^n$, and $\mathcal{A}$ mimics the reduction above and runs in time $O(n^4 \cdot 2^{1.08(nH(1/3))}) \ll 1.99^n$, which refutes \cref{conj:nuhampath_informal}.

\section{Preliminaries}
We denote the sets of integers and positive integers by $\Z$ and $\Z^+$, respectively. For a positive integer~$d$, by $[d]$ we denote the set $\{1,\ldots,d\}$.  
In this paper, we exclusively work with binary strings, i.e., vectors from $\{0, 1\}^d$ for some integer $d \geq 1$. For a vector $x\in \{0, 1\}^d$, we denote the Hamming weight of $x$ as $\norm{x}_0$. For $i\in[d]$, we use $x[i]$ to denote the $i$th coordinate of $x$. We say that $x,y\in\{0,1\}^d$ are orthogonal, and denote this by $x \perp y$, if $\langle x , y \rangle = \sum_{i \in [d]} x[i]\cdot y[i] =  0$, where the sum is taken over $\mathbb{Z}$.
For a subset $C \subseteq [d]$, $x |_C\in\{0,1\}^{|C|}$ is the vector $x$ restricted to the coordinates in $C$. We also use $\overline{C}$ to denote the complement of $C$.
We use the notation $\widetilde{O}(\cdot)$ to suppress polylogarithmic factors.

We will use the following simple bounds for binomial coefficients.
\begin{fact}
    For all $1 \leq k \leq n/2$, 
    \[\binom{n}{\leq k} \leq \left(\frac{en}{k}\right)^k\;.\]
\end{fact}

Let $H$ denote the binary entropy function where $H(p) = -p \log p - (1-p)\log (1-p)$ for $0<p<1$. Then,
\begin{fact}
    For all $1 \leq k \leq n$, 
    \[\binom{n}{k} \leq 2^{H(k/n)n}\;.\]
\end{fact}
We will also make use of the Chernoff bound, stated here for convenience.
\begin{lemma}[Chernoff Bound]
Let $X_1, \dots, X_n \in \{0, 1\}$ be independent random variables.
Let $X = X_1 + \dots + X_n$ and $\mu = E(X)$. For all $\beta \geq 0$, 
    \[\Pr[X \geq (1+\beta) \mu] \leq \exp\left(-\frac{\beta^2\mu}{(2+\beta)}\right)\;.\]
\end{lemma}

\subsection{Fine-Grained Complexity Conjectures}\label{sec:fgc}
In the $k$-SAT problem, given a $k$-CNF formula~$\varphi$, the task is to check if $\varphi$ has a satisfying assignment. The following hypothesis, introduced by Impagliazzo, Paturi, and Zane~\cite{IPZ98,IP99}, is now widely used in the field of fine-grained complexity for proving quantitative bounds on the complexity of algorithmic problems.
\begin{definition}[Strong Exponential Time Hypothesis (SETH)~\cite{IPZ98,IP99}]
    For every constant~$\eps > 0$ there exists $k \in \Z^+$ such that no deterministic algorithm solves $k$-SAT on formulas with $n$ variables in time $2^{(1 - \eps)n}$.
\end{definition}

In this work, we study the non-uniform version of SETH. Therefore, we begin with the definition of non-uniform algorithms.
\begin{definition}[Non-Uniform RAM]\label{def:non-uniform-alg}
A non-uniform RAM algorithm is a RAM algorithm $\A$ and an infinite family of advice strings $\{a_i\}_{i\in\Z^+}$. Such an algorithm solves a problem $\mathcal{P}$ in time $T$ and using $S$ bits of non-uniform advice if for all sufficiently large~$n$, (i)~the length of $|a_i|\leq S$, and (ii)~for every input $x$ of length~$n$, $\A(x, a_n)$ solves $\mathcal{P}$ on input~$x$ in time~$T$.
\end{definition}

Now, equipped with the definition of non-uniform algorithms, we are ready to state the Non-Uniform SETH.

\begin{conjecture}[Non-Uniform SETH (NUSETH)]\label{conj-nuseth}
For every constant $\eps \in (0,1)$ there exists $k \in \Z^+$ such that no non-uniform RAM algorithm solves $k$-SAT on formulas with~$n$ variables in time $2^{(1-\eps)n}$ using $2^{(1-\eps)n}$ bits of advice.
\end{conjecture}

\paragraph{On the choice of the computational model in NUSETH.}
Various non-uniform versions of SETH and related conjectures have been studied in previous works (see, e.g.,~\cite{GHLORS12,CEF12,CGIMPS16,ABGS21}). There are two natural choices for the computational model in the conjecture: non-uniform circuits and non-uniform RAM algorithms, and the former model would lead to a weaker conjecture.\footnote{Note that even for the extremely well-studied SETH conjecture, there is no clear agreement on the computational model. In particular, a $2^{n/2}$ RAM algorithm for SAT would not necessarily imply a Turing machine solving SAT in fewer than $2^n$ steps. Therefore, such an algorithm would refute SETH stated for RAM algorithms but not the Turing machines version of SETH.} For simplicity, in this paper, we choose to adapt the stronger version of the conjecture: no non-uniform RAM algorithm solves $k$-SAT in $2^{(1-\eps)n}$ for all~$k$. The point of this work is to show that when we construct a non-trivial data structure for {\OnlineOV} (or a number of related problems), we will obtain a faster-than-$2^n$ non-uniform algorithm for SAT. This will be great regardless of the definition of NUSETH!

Finally, we state the non-uniform version of the conjecture regarding the hardness of Hamiltonian Path. We say that a path in a graph is simple if it does not repeat any vertices. In the Hamiltonian Path problem (\HamPath), given a directed graph~$G$ on~$n$ vertices, the goal is to check if $G$ has a simple path that visits every vertex. While many special cases of \HamPath{} admit efficient algorithms (including Bj{\"o}rklund's celebrated algorithm for Undirected \HamPath{}~\cite{B10}), the best known algorithm for the general case of the problem is due to Bellman, and Held and Karp~\cite{B62, HK62}, and dates back 60 years. This algorithm runs in time $O(2^n \cdot \poly(n))$, and despite much effort, no faster algorithms are known even in the non-deterministic and non-uniform settings. 
Since progress has been hard on this problem we state the following non-uniform conjecture for Hamiltonian Path.

\begin{conjecture}[Non-Uniform \HamPath{} Conjecture]\label{conj-nuhampath} For every constant $\eps\in (0,1)$,  no non-uniform RAM algorithm solves {\HamPath} on graphs with $n$ vertices in time $2^{(1-\eps)n}$ using $2^{(1-\eps)n}$ bits of advice.
\end{conjecture}

We will use the following version of the dynamic programming algorithm for \HamPath{} from \cite{B62, HK62}.
\begin{proposition}[Implicit in~\cite{B62,HK62}]\label{prop:hk} There is a deterministic algorithm that for every $k:=k(n)$ outputs all simple paths of length at most~$k$ in a given $n$-vertex graph in time $\binom{n}{\leq k}\cdot\poly(n)$.
\end{proposition}
\begin{proof}
    For a set of vertices $S\subseteq[n]$ and a pair of distinct vertices $u,v\in S$, let $D[S, u, v]=1$ if there exists a simple path starting at the vertex $u$, passing through all vertices of~$S$, and ending at the vertex~$v$, and~$0$ otherwise. Each value $D[S, u,v]$ can be computed in linear time given the values of $D[S', u, v']$ for smaller sets~$S'$. Therefore, the na\"{i}ve dynamic programming algorithm computing $D[S,u,v]$ for all sets $S$ of size $|S|\leq k$ runs in time $\binom{n}{\leq k}\cdot\poly(n)$.
\end{proof}

\subsection{Data Structure Problems}

We define a data structure problem as a triplet $(\mathcal{D}, \mathcal{Q} , \mathcal{F})$, which corresponds to possible inputs $D\in\mathcal{D}$,  $D \in \bit^N$, a collection of queries $\mathcal{Q}$, and a function $\mathcal{F}: \mathcal{D} \times \mathcal{Q}\rightarrow \bit$ which encodes the answers to all the queries $q \in \mathcal{Q}$ for inputs from $\mathcal{D}$. 

\begin{definition}[Data Structures]\label{def:s-t-datastructure}
    An $(S, T)$-data structure for a problem $(\mathcal{D}, \mathcal{Q}, \mathcal{F})$ is a pair of deterministic algorithms $(\mathcal{P}, \mathcal{A})$ that solve the problem in two phases.

    \begin{description}
        \item{\bf Preprocessing phase.} In the first phase, the preprocessing algorithm $\P$ receives an input $D$, and preprocesses it into a data structure $\sigma\in\{0,1\}^S$ consisting of~$S$ bits.
        
        \item{\bf Query phase.} In the second phase, the query algorithm $\A$ runs in time $T$, and receives a query $q \in \mathcal{Q}$ and (query access to) the data structure~$\sigma$, and outputs $1$ if and only if $\mathcal{F}(D, q) = 1$.
    \end{description}

    In addition, we say that the data structure $(\mathcal{P}, \mathcal{A})$ has preprocessing time $T_\mathcal{P}$ if $\mathcal{P}$ runs in time at most $T_\mathcal{P}$.    
\end{definition}
We generally wish to simultaneously minimize all three computational parameters: the space $S$, the query time $T$, and the preprocessing time $T_\mathcal{P}$. However, in our lower bounds we will allow the preprocessing algorithm to be computationally unbounded, which only makes our results stronger. In the upper bounds, we will present algorithms with efficient preprocessing.

Our lower bounds will hold even for stronger ``batch'' versions of the data structure problems, where the query algorithm is given a number of queries simultaneously and is asked to solve them all together. For a number of problems, the batch versions are easier to solve. (A~prominent example is the Matrix-Vector Multiplication problem, where the batch version, Matrix-Matrix Multiplication, is known to have efficient algorithms). 
\begin{definition}[Batch Data Structures]\label{def:batch_ds}
        A batch data structure for a problem $(\mathcal{D}, \mathcal{Q}, \mathcal{F})$ solving $M$ queries in space~$S$ and time~$T$ per query is a pair of deterministic algorithms $(\mathcal{P}, \mathcal{A})$ that solve the problem in two phases.

    \begin{description}
        \item{\bf Preprocessing phase.} In the first phase, the preprocessing algorithm $\P$ receives an input $D$, and preprocesses it into a data structure $\sigma\in\{0,1\}^S$ consisting of~$S$ bits.
        
        \item{\bf Query phase.} In the second phase, the query algorithm $\A$ runs in time $M\cdot T$, and receives $M$ queries $q_1,\ldots,q_M \in \mathcal{Q}$ and (query access to) the data structure~$\sigma$, and outputs $a_1,\ldots,a_M$ such that $a_i=1$
        if and only if $\mathcal{F}(D, q_i) = 1$.
    \end{description}

    In addition, we say that the data structure $(\mathcal{P}, \mathcal{A})$ has preprocessing time $T_\mathcal{P}$ if $\mathcal{P}$ runs in time at most $T_\mathcal{P}$.   
\end{definition}

We will also be interested in average-case data structures. We are particularly interested in data structures which are efficient over some distribution of inputs $\mathcal{B}$.

\subsection{Online Orthogonal Vectors}
In this section, we formally define the Orthogonal Vectors problem (OV) and its data structure variant $\OnlineOV$. 

\begin{definition}[Orthogonal Vectors ($\OV_d$)]
Given two sets of vectors $A, B \subseteq \{0, 1\}^d$ of size $|A|=|B| = n$, determine if there exists $x \in A$ and $y \in B$ such that $x \perp y$. 
\end{definition}

While $\OV_d$ admits a trivial algorithm running in time $O(n^{2}d)$, no truly sub-quadratic algorithm is known.
In fact it was shown in \cite{W05} that refuting the following conjecture would break SETH.

\begin{conjecture}[OV Conjecture]\label{conj-ovc} There is no $\Ot(n^{2-\eps})$-time algorithm solving $OV_{c\log n}$ for all constant $c>0$.
\end{conjecture}

We will be interested in the data structure version of Orthogonal Vectors which we define below.

\begin{definition}[Online Orthogonal Vectors  ($\OnlineOV_{n,d}$)]
For a dimension parameter $d:=d(n)$, $\OnlineOV_{n,d}$  is a data structure problem $(\mathcal{D}, \mathcal{Q}, \mathcal{F})$ where:
\begin{itemize}
    \item $\mathcal{D}=\{D\colon D \text{ is a set of vectors }\{x_1, \dots, x_n\}, \text{ where } x_i \in\{0,1\}^{d}\}$.
    \item $\mathcal{Q}$ is the set of all vectors in $\bit^d$.
    \item  $\mathcal{F}(D, q) = 1$ if and only if there exists $x_i \in D$ such that $x_i \perp q$.
\end{itemize}

In the special case where for every $x_i \in \mathcal{D}$, $\norm{x_i}_0 = c$, we call this problem \emph{$c$-Sparse-$\OnlineOV_{n, d}$}.
\end{definition}

We will be primarily interested in two parameter regimes---when $d = c \log n$ for a constant $c > 0$ ($\OnlineOV_{n, c \log n}$) and when $d = n^{\eps}$ for a small constant $0 < \eps < 1$ ($\OnlineOV_{n, n^{\eps}}$).

We will also be interested in the average-case version of $\OnlineOV_{n, d}$. We are primarily interested in the case where the input vectors are drawn i.i.d.~from the $p$-biased distribution of $n$-bit vectors.

\begin{definition}[Average-Case Online Orthogonal Vectors  ($\OnlineOV_{n,d}$)]
    An average-case $(S, T)$-data structure for $\OnlineOV_{n, d}$ over a distribution of inputs $\mathcal{B}$ is a pair of deterministic algorithms $(\mathcal{P}, \mathcal{A})$ solving each input~$D$ of the problem in space~$S_D$, worst-case query time $T_D$, and preprocessing time $T_{\mathcal{P},D}$ such that  $\mathbb{E}_{D \sim \mathcal{B}}[S_D] \leq S$, $\mathbb{E}_{D \sim \mathcal{B}}[T_D] \leq T$. Furthermore, the preprocessing time is defined as $\mathbb{E}_{D \sim \mathcal{B}}[T_{\mathcal{P}, D}]$.
\end{definition}

\subsection{Other Data Structure Problems} \label{section:ds_reductions}

We will work with the following related data structure problems.
\begin{itemize}
    \item $\PM_{n,d}$: Preprocess~$n$ vectors $(x_1,\ldots,x_n)$ from $\{0,1\}^d$, and for a query vector~$q\in\{0,1, *\}^d$, output~$1$ if and only if there exists~$i$ such that for all~$j$: $q[j]=*$ or $q[j]=x_i[j]$.

    \item $\SubsetQuery_{n,d}$: Preprocess $n$ sets $S_1, \hdots S_n \subseteq [d]$, and for a query $q \subseteq [d]$ output $1$ if and only if there exists an $i$ such that $q \subseteq S_i$.

    \item $\ContainmentQuery_{n,d}$: Preprocess $n$ sets $S_1, \hdots S_n \subseteq [d]$, and for a query $q \subseteq [d]$ output $1$ if and only if there exists an $i$ such that $S_i \subseteq q$.

    \item $\ORS_{n,d,u}$: 
    Let $u$ be an integer. Preprocess $n$ points $(x_1, \dots, x_n)$ from $[u]^d$, and for a query vector $q=(\ell_1,r_1,\ldots,\ell_d,r_d)$ that represents a rectangle in $[u]^d$ output $1$ if there is an $x_i$ contained in the rectangle, i.e., if there exists $i$ such that for all $j \in [d]$: $\ell_j\leq x_i[j] \leq r_j$.

    \item $(1+\eps)$-$\ORS_{n,d,u}$: A generalization of $\ORS_{n, d, u}$ where every query rectangle is approximated by a larger rectangle containing the query rectangle where each side is blown up by a factor of $(1 + \eps)$. %

    \item $\DNFEval_{n,d}$:
    Let $\phi$ be a DNF formula with $n$ clauses and $d$ variables. Preprocess $\phi$, and for a query $q \in \bit^d$, output $1$ if and only if $\phi(q) = 1$.

    \item $\OnlineIP_{n, d, k}$: Preprocess $n$ vectors $(x_1, \hdots, x_n)$ from $\bit^d$, and for a query vector $q \in \bit^d$, output $1$ if and only if there exists $i$ such that $\langle x_i, q\rangle = k$. 

    \item $(1+\eps)$-$\ANN_{n, d, u, p}$: Let $u$ be an integer. Preprocess $n$ points $(x_1, \dots, x_n)$ from $[u]^d$, and for a query vector $q \in [u]^d$ output $x_i$\footnote{We abuse notation here and describe a search problem instead of a decision problem.} such that $\norm{x_i - q}_{p} \leq (1 + \eps) \dot \min_{j}\norm{x_j - q}_p$. 
\end{itemize}

\begin{lemma}[\cite{I98,CIP02, AWY14, CW19}]\label{lem:ov_reduces_to_everything}
If there is an $(S,T)$-data structure with preprocessing time $T_{\mathcal{P}}$ for any of the following problems, then there is an $(O(S), O(T))$-data structure with preprocessing time $O(T_{\mathcal{P}})$ for $\OnlineOV_{n, d}$.
    \begin{align*}& \PM_{n, d}, \; \SubsetQuery_{n, d}, \; \ContainmentQuery_{n, d}\;, \\
    & \ORS_{n, d, 1}, \; \DNFEval_{n,d}, \; \OnlineIP_{n, d, 0}\;, \\
     & (1+\eps)\text{-}\ANN_{n, d, p, \infty}\;.
    \end{align*}
\end{lemma}

\begin{lemma}[\cite{CIP02, AWY14}]\label{lem:everything_is_ov}
If there is an $(S,T)$-data structure with preprocessing time $T_{\mathcal{P}}$ for $\OnlineOV_{n, d}$, then there is an $(O(S), O(T))$-data structure with preprocessing time $O(T_{\mathcal{P}})$ for each of the following problems.
    \begin{align*} 
     & \PM_{n, d/2}, \;\DNFEval_{n, d/2}, \; \SubsetQuery_{n, d}, \\ & \ContainmentQuery_{n, d}, \; (1+\eps)\text{-}\ORS_{n, O(d\eps/\log u), u}\; .
    \end{align*}
\end{lemma}

We remark that the reductions of \cref{lem:ov_reduces_to_everything,lem:everything_is_ov} hold for the batch versions of the problems as well.

We will also use the following more complicated and randomized reduction from \ORS{} to \OnlineOV{}.\footnote{We say that $(\P, \A)$ is a randomized $(S, T)$-data structure if both the preprocessing and query algorithms $\P$ and $\A$ are allowed to use randomness.}

\begin{lemma}[\cite{CIP02}]
    If there is an $(S, T)$-data structure for $\OnlineOV_{n, d'}$ where $d' = O(d^2\log^2u)$, then there is an $(O(S), O(T))$-randomized data structure for $\ORS_{n,d, u}$.
\end{lemma}

\section{Algorithms}\label{sec:algs}

In this section, we provide new average-case and worst-case deterministic data structures for $\OnlineOV_{n, d}$. Using known reductions from other data structure problems to $\OnlineOV$ (see \cref{section:ds_reductions}), we also find new data structures for those problems. 
\subsection{Average-Case Algorithm}
We begin this section by presenting a simple algorithm for the average-case version of the \OnlineOV{} problem. In~\cref{def:pseudorandom}, we identify a specific property of random instances of \OnlineOV{} that enables the design of an efficient algorithm, as described in \cref{thm:average}. Then, in \cref{sec:worst-case}, we extend these ideas to develop an algorithm for the worst-case version of the problem.

\begin{definition}\label{def:pseudorandom}
    Let $X \subseteq \bit^d$ be a set of $n = |X|$ vectors of dimension $d$.
    We say that $X$ is \textit{$(m,t)$-pseudorandom} if for all subsets of vectors $X' \subseteq X$ of size $|X'| = m$ and all subsets of coordinates $C \subseteq [d]$ of size $|C| = t$, there exists some $x \in X'$ with $x|_C \not= \vec{0}$.
\end{definition}

\noindent We show that a random input to the \OnlineOV{} problem is pseudorandom with high probability.

\begin{lemma}\label{lem:pseudorandom}
    Let $n$, $d := d(n)$, $m \leq n$, and $t \leq d$ all be positive integers.
    Let $X \subseteq \bit^d$, $|X| = n$, with each entry of the vectors of $X$ sampled to be zero independently with probability $p \in (0,1)$.
    Furthermore, assume $m \geq 3np^t$.
    Then
    \begin{equation*}
        \Pr_X[X \text{ is } (m,t)\text{-pseudorandom}] \geq 1 - \binom{d}{t}\cdot \exp(-m/3).
    \end{equation*}
\end{lemma}
\begin{proof}
    For $C \in \binom{[d]}{t}$, let $Y_C = \{x \in X : x|_C=\vec{0}\}.$
    Then by definition, 
    \begin{equation*}
        \Pr_X[X \text{ is not } (m,t)\text{-pseudorandom}] = \Pr_X[\exists \, C \in \binom{[d]}{t} : |Y_C| \geq m].
    \end{equation*}
    Furthermore, by the independence of the entries of the vectors of $X$, a union bound gives us
    \begin{equation*}
        \Pr_X[\exists \, C \in \binom{[d]}{t} : |Y_C| \geq m] \leq \binom{d}{t} \cdot \Pr_X[|Y_{C'}| \geq m],
    \end{equation*}
    where $C'$ is some arbitrarily chosen subset of $[d]$.
    Note that $\E[|Y_{C'}|] = n p^t$.
    Thus, by a Chernoff bound with $\beta = \frac{m}{np^t}-1 \geq 2$,
    \begin{equation*}
        \Pr_X[|Y_{C'}| \geq m = np^t\cdot(1+\beta)] \leq \exp(-\beta^2np^t/(2+\beta)) \leq \exp(-\beta np^t/2).
    \end{equation*}
    Because $\beta \geq \frac{2m}{3np^t}$, we have $\Pr_X[|Y_{C'}| \geq m] \leq \exp(-\beta np^t/2) \leq \exp(-m/3)$.

    Altogether, this gives us
    \begin{equation*}
        \Pr_X[X \text{ is not } (m,t)\text{-pseudorandom}] \leq \binom{d}{t}\cdot \exp(-m/3). \qedhere
    \end{equation*}
\end{proof}

Now we show that when the input is pseudorandom, our algorithms can solve \OnlineOV{} efficiently.
Additionally, even when the input is not pseudorandom, our algorithms do not use too much time and space.
This gives us good bounds on the expected space and query time. We are interested in the case where the inputs $X \subseteq \{0, 1\}^d, |X| = n$ are sampled from the distribution $\mathcal{B}_p$ where each entry in every vector is set to $0$ independently with probability $p$.

\begin{algorithm}[!ht]
\DontPrintSemicolon
\KwIn{A set of $n$ vectors $X \subseteq \bit^{d}$, pseudorandom parameter $1 \leq t \leq d$}
\KwOut{An advice string $\sigma=\hat{\sigma} \concat \sigma'$}

$\hat{\sigma} \gets$ empty string \;
\ForEach(\tcp*[f]{Bitmap for all sparse queries}){$q \in \bit^{d}$ with $\norm{q}_0 < t$}{
    \If{$\exists \, x \in X$ such that $x \perp q$}{
        $\hat{\sigma} \gets \hat{\sigma} \concat 1$ \;
    }
    \Else{
        $\hat{\sigma} \gets \hat{\sigma} \concat 0$ \;
    }
}

$\sigma' \gets$ empty string \;
\ForEach(\tcp*[f]{Store a list of candidate vectors for each $C$}){$C \in \binom{[d]}{t}$}{
    $Y_C \gets \{x \in X : x|_C = \vec{0}\}$ \;
    $\sigma' \gets \sigma' \concat Y_C$ \;
}

\Return $\hat{\sigma} \concat \sigma'$ \;
\caption{\texttt{AverageOVPre}}\label{alg:AverageOVPre}
\end{algorithm}

\begin{algorithm}[!ht]
\DontPrintSemicolon
\KwIn{A query $q \in \bit^{d}$, (query access to) an advice string $\sigma=\hat{\sigma} \concat \sigma'$,  pseudorandom parameter $1 \leq t \leq d$}
\KwOut{A bit indicating whether there exists some $x \in X$ such that $x \perp q$}

\If{$\norm{q}_0 < t$}{
    $b \gets$ answer for $q$ from $\hat{\sigma}$ \tcp*[f]{Look up answer in bitmap} \;
    \Return $b$ \;
}

$C \gets$ arbitrary subset of $t$ coordinates for which $q|_C = \vec{1}$ \;
$Y \gets$ corresponding $Y_C$ in $\sigma'$ \;
\For(\tcp*[f]{Check corresponding candidate vectors}){$x \in Y$}{
    \If{$x \perp q$}{
        \Return $1$ \;
    }
}

\Return $0$\;
\caption{\texttt{AverageOVOnl}}\label{alg:AverageOVOnl}
\end{algorithm}

\begin{theorem}\label{thm:average}
    Let $n\geq 1$, $d:=d(n)$, $0 < \eps < 0.99$. Furthermore assume $\log_{1/p}(2n^\eps) \leq d \leq 2^{n^{1-\eps'}}$ for some $\eps' > \eps$. Then for large enough $n$, the pair of deterministic algorithms 
    $(\emph{\texttt{AverageOVPre}}(\cdot, t), \emph{\texttt{AverageOVOnl}}(\cdot, \cdot, t))$ is an average-case $(S, T)$-data structure for $\OnlineOV_{n,d}$ over $\B_p$ with preprocessing time $T_p$ where
    \begin{align*}
        T &\leq O(n^{1-\eps} d), & 
        S &\leq \binom{d}{\leq t} \cdot n^{1-\eps} d, &\text{and} &&
        T_p  &\leq  O\left(\binom{d}{ \leq t} \cdot nd\right), 
    \end{align*}
    with $t = \log_{1/p}(6n^\eps)$.
    Furthermore, for each $D \in \mathrm{Supp}(\mathcal{B}_p)$, $S_D \leq  \binom{d}{ \leq t} \cdot nd$.
\end{theorem}

Before proving the theorem, we first restate it with $p=1/2$ and where the savings in query time is some fixed factor $1 \leq \alpha \leq n$.

\begin{corollary}\label{cor:compare_cip_avg}
    Let $n\geq 1$, $d:=d(n)$, $1 \leq \alpha := \alpha(n) \leq n$. Furthermore assume $\log(2\alpha) \leq d \leq 2^{n^{1-\eps'}}$ for some $\eps' > \log(\alpha)/\log (n)$. Then there is an average-case  $(S, T)$-data structure for $\OnlineOV_{n,d}$ over $\mathcal{B}_{1/2}$ with preprocessing time $T_p$ where
    \begin{align*}
            T &= O(nd/\alpha), & S &\leq 2^{O(\log(\alpha)\log(d/\log \alpha))} \cdot nd/\alpha, 
            && \text{ and} & \;T_p &=  O(2^{O(\log(\alpha)\log(d/\log \alpha))} \cdot nd),
    \end{align*}

\end{corollary}
\begin{proof}
    We apply \cref{thm:average} with $\eps = \log(\alpha)/\log(n)$ and use the bound
    \begin{equation*}
        \binom{d}{\leq t}  \leq 2^{t \log(ed/t)} = 2^{\log(6 \alpha)\log\left(\frac{ed}{\log(6\alpha)}\right)}. \qedhere
    \end{equation*}
\end{proof}

We are also particularly interested in the case where $d = c \log n$ for some constant $c > 1$.

\begin{corollary}\label{cor:chan_comparison_avg}
    Let $n \geq 1$, $c > 1, \delta > \frac{2e\log c}{c}$, $d:= c\log n$. 
    Then there is an average-case  $(S, T)$-data structure for $\OnlineOV_{n,d}$ over $\mathcal{B}_{1/2}$ with preprocessing time $T_p$ where
    \begin{equation*}
            T \leq n^{1-\Omega(\frac{1}{\log c})} \quad \text{and} \quad 
            S, \; T_p \leq \tilde{O}(n^{1+\delta}) \;.
    \end{equation*}
\end{corollary}
\begin{proof}
    We apply \cref{thm:average} with $\eps = \frac{\delta}{4 \log c}$.
    Then $t = \log(6n^{\eps}) \leq \frac{\delta\log n}{2\log c}$ and $\binom{d}{\leq t} \leq 2^{\frac{\delta\log n}{2 \log c} \cdot \log(\frac{2ec \log c}{\delta})} < n^\delta$.
\end{proof}

We now return to the proof of \cref{thm:average}.

\begin{proof}[Proof of \cref{thm:average}]
    First, to see that the algorithms are correct, note that for ``sparse'' queries, i.e.~queries $q \in \bit^d$ with $\norm{q}_0 < t$, \texttt{AverageOVPre} simply stores the correct answer for $q$, and $\texttt{AverageOVOnl}$ looks this answer up.
    For all other queries, \texttt{AverageOVPre} writes down a list of ``candidate'' vectors $Y_C = \{x \in X : x|_C = \bf{0}\}$ for every $C \subseteq [d]$ with $|C| = t$.
    \texttt{AverageOVOnl} then finds the first $t$ non-zero coordinates $C \subseteq [d]$ of $q$, checks to see if any of the vectors in $Y_C$ are orthogonal to $q$, and returns this answer.
    Because any vector in $X$ that is orthogonal to $q$ must appear in $Y_C$, this answer will be correct.

    Now we will bound the worst-case space and time for our algorithms.
    For space, note that a single bit is stored for every sparse query, which is $\binom{d}{<t}$ bits in total.
    Constructing this table takes at most $O(\binom{d}{<t}\cdot nd)$ preprocessing time, and looking up an answer takes at most $O(dt)$ query time.
    For dense queries, \texttt{AverageOVPre} stores $\binom{d}{t}$ many lists of candidate vectors.
    Let $M$ be the number of vectors in the largest of such lists.
    Then storing these lists takes at most $\binom{d}{t}\cdot Md$ space, and constructing these lists takes $O(\binom{d}{t} \cdot n d)$ preprocessing time.
    To check a list against a dense query vector, \texttt{AverageOVOnl} takes at most $O(Md)$ query time.
    Thus, overall we have 
    \begin{align}
        S \leq \binom{d}{\leq t}\cdot Md \leq \binom{d}{\leq t}\cdot nd, && T_p \leq O\left(\binom{d}{\leq t}\cdot nd\right), && \text{and} && T \leq O(d \cdot \max(M, t)). \label{eq:avgcase-worst-bds}
    \end{align}

    Now we bound $\mathbb{E}[M]$.
    Let $m = n^{1-\eps}/2$ and notice that $m = 3np^t$.
    Thus, by \cref{lem:pseudorandom}, $X$ is $(m,t)$-pseudorandom (i.e.~$M \leq m$) with probability at least $1-\binom{d}{t}\cdot \exp(-m/3)$. 
    We will show that for large enough $n$, $\binom{d}{t} \cdot \exp(-m/3) \leq m/n$. 
    This implies
    \begin{align*}
        \mathbb{E}[M] &\leq \left(1-\binom{d}{t}\cdot \exp(-m/3)\right) \cdot m + \binom{d}{t}\cdot \exp(-m/3) \cdot n \\
        &\leq m + \binom{d}{t}\cdot \exp(-m/3) \cdot n \\
        & \leq 2m \;= n^{1-\eps}.
    \end{align*}
    Our result follows by combining the above observation with \cref{eq:avgcase-worst-bds}.
    Finally, to see that $\binom{d}{t} \cdot \exp(-m/3) \leq m/n$, note that because $d \leq 2^{n^{1-\eps'}}$ and $t=\log_{1/p}(6n^\eps)$, we have $t + t\ln(d/t) + \ln(2n^\eps) \leq n^{1-\eps}/6$ for large enough $n$.
    This implies that $(ed/t)^t \cdot \exp(-n^{1-\eps}/6) \leq n^{-\eps}/2$, and so
    \begin{align*}
        \binom{d}{t} \cdot \exp(-m/3) & \leq (ed/t)^t \cdot \exp(-n^{1-\eps}/6) \\
        &\leq n^{-\eps}/2 =  m/n,
    \end{align*}
    completing the proof.
\end{proof}

\subsection{Worst-Case Algorithm}\label{sec:worst-case}

To extend the approach to a worst-case algorithm, we must handle the parts of the inputs that are not pseudorandom in the sense of \cref{def:pseudorandom}, i.e.~the large subsets of vectors that share many common zeros.
These parts are identified by \cref{alg:partition} which partitions the input vectors into a single pseudorandom part and multiple ``structured'' parts, each of which contains vectors that share many common zeros.
This algorithm is used repeatedly as a subroutine in our preprocessing algorithm.

\begin{algorithm}[!ht]
\DontPrintSemicolon
\KwIn{A set of vectors $X \subseteq \bit^d$, pseudorandom parameters $m, t \in \Z^+$.}
\KwOut{A partition $X = X' \sqcup X_1 \sqcup \dots \sqcup X_k$ and sets $S_1, \dots, S_k$ with  $|X_i| = m$ and $|S_i| = t$ such that $X'$ is $(m, t)$-pseudorandom and for every $i \in [k]$ and every $x \in X_i$, $x|_{S_i} = \vec{0}$.}

$k \gets 0$ \;

\ForEach{$S \in \binom{[d]}{t}$}{ \label{algline:for-S}
    $X_S \gets \{x \in X : x|_S = \vec{0}\}$\; \label{algline:X_S}
    \While{$|X_S| \geq m$}{ \label{algline:large_X_S_condition} 
        $k \gets k + 1$\;
        $X_k \gets$ first $m$ elements of $X_S$ \label{algline:X_k-def}\;
        $S_k \gets S$\;
        $X \gets X \setminus X_k$\;
        $X_S \gets \{x \in X : x|_S = \vec{0}\}$ \label{algline:X_S2}\;
    }
}

\Return $(X, X_1, \dots, X_{k})$, $(S_1, \dots, S_k)$\;

\caption{\texttt{PseudorandomPartition}}\label{alg:partition}
\end{algorithm}

\begin{lemma} \label{lem:part_alg_properties}
    Let $X \subseteq \bit^d$ be a set of $n = |X|$ vectors and $m, t \in \Z^+$ be parameters.
    Let $\emph{\texttt{PseudorandomPartition}}(X, m, t) = ((X', X_1, \dots, X_k), (S_1, \dots, S_k))$.
    Then the following are true:
    \begin{itemize}
        \item $X'$ is $(m, t)$-pseudorandom,
        \item for every $i \in [k]$, $|X_i| = m$, and $|S_i|=t$, and
        \item for every $i \in [k]$ and $x \in X_i$, $x|_{S_i} = \vec{0}$.
    \end{itemize}
    Moreover, $\emph{\texttt{PseudorandomPartition}}$ runs in time $O(\binom{d}{t} \cdot nt)$.
\end{lemma}
\begin{proof}
    First, by the definition of $X_S$ in line \ref{algline:X_S} (and line \ref{algline:X_S2}), we have that for every $i \in [k]$ and every $x \in X_i$, $x|_{S_i} = \vec{0}$.
    Furthermore, the definitions of $S$ and $X_k$ in lines \ref{algline:for-S} and \ref{algline:X_k-def}, respectively, ensure that for every $i \in [k]$, $|X_i|=m$ and $|S_i|=t$.

    To see that $X'$ is $(m, t)$-pseudorandom, assume we have some subset $\hat{X} \subseteq X$ with $|\hat{X}| \geq m$ and some set of coordinates $\hat{S} \subseteq [d]$ of size $|\hat{S}| = t$ such that every $x \in \hat{X}$ has $x|_{\hat{S}} = \vec{0}$.
    Then if $\hat{X}$ appears as a subset of $X_S$ in line $\ref{algline:X_S}$ when $S = \hat{S}$, then $\hat{X}$ will be removed from $X$ and therefore not appear as a subset of $X'$.
    The only way $\hat{X}$ could not appear as a subset of $X_S$ when $S = \hat{S}$ is if vectors in $\hat{X}$ were already removed from $X$ in prior iterations of the loop.
    In this case too, $\hat{X}$ would not appear as a subset of $X'$.
    Therefore, $X'$ is $(m, t)$-pseudorandom. %

    Lastly, to analyze the runtime, we iterate through the loop $\binom{d}{t}$ times.
    In each iteration, we run through the at most $n$ many vectors in $X$ once to check if $x|_S = \vec{0}$.
    This takes $\binom{d}{t}\cdot nt$ time.
    We enter the While loop at most $n / m$ times, and each run takes $O(m)$ time.
    Therefore, overall \texttt{PseudorandomPartition} runs in time $O(\binom{d}{t} \cdot nd)$.
\end{proof}

We now design a recursive algorithm that, at each step, identifies a pseudorandom part of the input (using \cref{alg:partition}) and partitions the remaining input into several structured parts. 
The algorithm then proceeds by solving the pseudorandom and structured parts separately. 
For the pseudorandom part---similar to the approach in the average-case algorithm from \cref{thm:average}---we store short lists of relevant inputs for each dense query vector.
For each structured part, we observe that the fact that all vectors share many zeros in common allows us to reduce the dimensionality of the problem and recursively apply our algorithm to the resulting sub-instances.

We present the preprocessing algorithm in \cref{alg:OVPre} and the query algorithm in \cref{alg:OVOnl}. We then give a rigorous analysis of the algorithms in \cref{thm:worst_case_algorithm_parameterized}.

\begin{algorithm}[!ht]
\DontPrintSemicolon
\KwIn{A set of $n$ vectors $X \subseteq \bit^{d}$, integer $i \geq 1$.}
\KwOut{An advice string $\sigma = \hat{\sigma} \concat \sigma' \concat \sigma_1 \concat \cdots \concat \sigma_k$.}

\If{$n = 1$} {
    $\sigma \gets \{x \in X\}$\;
    \Return $\sigma$ \;
}

\If{$i = 1$}{
    $\sigma \gets$ empty string \;
    \ForEach(\tcp*[f]{Bitmap for all queries}){$q \in \bit^{d}$}{
        \If{$\exists \, x \in X$ such that $x \perp q$}{
            $\sigma \gets \sigma \concat 1$ \;
        }
        \Else{
            $\sigma \gets \sigma \concat 0$ \;
        }
    }
    \Return $\sigma$ \;
}

$m \gets n^{1-\frac{1}{i}}$ \;
$t \gets \frac{d}{i}$ \;

$\hat{\sigma} \gets$ empty string \;

\ForEach(\tcp*[f]{Bitmap for all sparse queries}){$q \in \bit^{d}$ with $\norm{q}_0 < t$}{
    \If{$\exists \, x \in X$ such that $x \perp q$}{
        $\hat{\sigma} \gets \hat{\sigma} \concat 1$ \;
    }
    \Else{
        $\hat{\sigma} \gets \hat{\sigma} \concat 0$ \;
    }
}

$(X', X_1, \dots, X_k), (C_1, \dots, C_k) \gets \texttt{PseudorandomPartition}(X, m, t)$ \;

$\sigma' \gets$ empty string \;

\ForEach(\tcp*[f]{List of candidates for pseudorandom part}){$C \in \binom{[d]}{t}$}{
    $Y_C \gets \{x \in X' : x|_C = \vec{0}\}$ \label{alg:ovpre_YC}\;
    $\sigma' \gets \sigma' \concat Y_C$ \;
}

\For(\tcp*[f]{Recursively solve all structured parts}){$j \in [k]$}{
    Restrict each $x \in X_j$ to coordinates $[d] \setminus C_j$ \;
    $\sigma_j \gets C_j \concat \texttt{OVPre}(X_j, i-1)$ \;
}

\Return $\sigma = \hat{\sigma} \concat \sigma' \concat \sigma_1 \concat \cdots \concat \sigma_k$ \;

\caption{\texttt{OVPre}}\label{alg:OVPre}
\end{algorithm}

\begin{algorithm}[!ht]
\DontPrintSemicolon
\KwIn{A query $q \in \bit^{d}$, (query access to) an advice string $\sigma = \hat{\sigma} \concat \sigma' \concat \sigma_1 \concat \cdots \concat \sigma_k$, integer $1 \leq i \leq d$.}
\KwOut{A bit indicating whether there exists some $x \in X$ such that $x \perp q$.}
\If{$n = 1$ }{
    $x \gets \sigma$ \;
    \Return $x \perp q$\;
}

\If{$i = 1$}{
    $b \gets$ answer for $q$ from $\sigma$ \tcp*[f]{Look up answer in bitmap} \;
    \Return $b$ \;
}

$t \gets \frac{d}{i}$ \; 

\If{$\norm{q}_0 < t$}{
    $b \gets$ answer for $q$ from $\hat{\sigma}$ \tcp*[f]{Look up answer in bitmap} \;
    \Return $b$ \;
}

$C \gets$ arbitrary subset of $t$ coordinates for which $q|_C = \vec{1}$ \;
$Y \gets$ corresponding $Y_C$ in $\sigma'$ \;
\For(\tcp*[f]{Check candidates in pseudorandom part}){$x \in Y$}{
    \If{$x \perp q$}{
        \Return $1$ \;
    }
}
\For(\tcp*[f]{Recursively solve on structured parts}){$j \in [k]$}{
    $b \gets \texttt{OVOnl}(q|_{\overline{C_j}}, \sigma_j, i-1)$ \;
    \If{$b = 1$}{
        \Return $1$ \;
    }
}

\Return $0$ \tcp*[f]{No orthogonal vectors found in any part} \;

\caption{\texttt{OVOnl}}\label{alg:OVOnl}
\end{algorithm}

\begin{theorem} \label{thm:worst_case_algorithm_parameterized}
    For every $n\geq 1$, $d:=d(n)$, and an integer parameter $1 \leq i \leq d$, the pair of deterministic algorithms 
    $(\emph{\texttt{OVPre}}(\cdot,i), \emph{\texttt{OVOnl}}(\cdot, \cdot, i))$ is an $(S, T)$-data structure with preprocessing time $T_p$ for $\OnlineOV_{n,d}$ where
    \begin{align*}
        T \leq 2idn^{1-1/i}, &&
        S \leq \binom{d}{\leq d/i}\cdot idn^{1-1/i},  &&  \text{and}  &&
        T_p \leq \binom{d}{\leq d/i}\cdot idn.
    \end{align*}
\end{theorem}

Before proving \cref{thm:worst_case_algorithm_parameterized}, we instantiate it with $i=\log{n}/(\eps\log{n}+\log\log{n})$ to obtain a query time of $T \leq n^{1-\eps}d$.

\begin{corollary}\label{cor:cip_comparison}
    Let $n \geq 2$, $d \geq \log n / 2$, and $\frac{\log \log n}{\log n} \leq \eps < 1/2$.
    There is a $(S, T)$-data structure with preprocessing time $T_p$ for $\OnlineOV_{n,d}$ where 
    \begin{align*} T &\leq n^{1-\eps}d, & S & \leq 2^{O(\eps d\log(1/\eps))} \cdot n^{1-\eps}d, && \text{and}
    &&  T_p \leq 2^{O(\eps d \log(1/\eps))}\cdot n^{1+3\eps}d.
    \end{align*}
\end{corollary}

\begin{proof}
    Apply \cref{thm:worst_case_algorithm_parameterized} with $i = \left\lfloor \frac{\log n}{\eps\log n+\log \log n} \right\rfloor \leq \frac{\log n}{2}$. Then, 
    \begin{equation*}
        T \leq 2id \cdot n^{1-1/i} \leq d \log n \cdot \frac{n}{n^{(\eps \log n + \log \log n)/\log n}} = n^{1-\eps}d.
    \end{equation*}
    Note that $i \geq \frac{\log n}{2(\eps\log n+\log \log n)}$ and $\log \log n/\log n \leq \eps$. Therefore, 
    \begin{align*}
        S &\leq \binom{d}{\leq d/i}\cdot idn^{1-1/i}\\
         & \leq 2^{\frac{d\log(ei)}{i}} \cdot \frac{n^{1-\eps} d}{2} \\
         & \leq 2^{4d\eps\log\left(\frac{e}{\eps}\right)} \cdot \frac{n^{1-\eps} d}{2} \\
         & = 2^{O\left(\eps d \log(1/\eps)\right)} \cdot n^{1-\eps} d.
    \end{align*}
    Likewise,
    \begin{align*}
        T_p &\leq \binom{d}{\leq d/i}idn \leq 2^{O(\eps d \log(1/\eps)}\cdot n^{1-\eps}d \cdot n^{1/i} \leq 2^{O(\eps d \log(1/\eps))}\cdot n^{1+3\eps}d \; . \qedhere
    \end{align*}
\end{proof}

Additionally, we are interested in our algorithm's performance for the special case of $d=c\log{n}$ for constant $c$.

\begin{corollary} [$d = c \log n$] \label{cor:chan_comparison}
    Consider any $n \geq 2$, $c \geq 2$, and $\delta \geq 2e\log c/c$. There is a $(S, T)$-data structure with preprocessing time $T_p$ for $\OnlineOV_{n,c \log n}$ where
    \begin{align*}
       T = n^{1-\Omega\left(\frac{1}{c\log c}\right)}   \quad \text{and} \quad S, \;T_p = \tilde{O}(n^{1+\delta}) \;.
    \end{align*}
\end{corollary}

\begin{proof}
    Apply \cref{thm:worst_case_algorithm_parameterized} with $i = 2c\log c/\delta$. Then
    \begin{align*}
        T & \leq 2idn^{1-1/i} = n^{1-\Omega\left(\frac{1}{c\log c}\right)}, \\
        S &\leq \binom{d}{\leq d/i}\cdot idn^{1-1/i} 
        \leq \tilde{O}(2^{\frac{\delta\log n}{2 \log c} \cdot \log(2e c \log c / \delta)} \cdot n)
        = \tilde{O}(n^{1+\delta}), \text{ and}\\
        T_p &\leq \binom{d}{\leq d/i}\cdot idn 
        \leq \tilde{O}(2^{\frac{\delta\log n}{2 \log c} \cdot \log(2e c \log c / \delta)} \cdot n)
        = \tilde{O}(n^{1+\delta}) \;.\qedhere
    \end{align*}
\end{proof}

We now proceed with the proof of \cref{thm:worst_case_algorithm_parameterized}.

\begin{proof}[Proof of \cref{thm:worst_case_algorithm_parameterized}]
    We proceed by induction on $n$, $i$, and $d$. Our induction hypothesis assumes that for each $i' < i, n' \leq n$ and $d' \leq d$, the algorithms satisfy the stated bounds.
    At each step of the induction, we have $i' = i-1$, $d' = d-d/i$, and $n' = n^{1-1/i}$.
    Thus, this preserves the invariant $i' \leq d'$, so the only necessary base cases are when $n=1$ or $i=1$.
    \paragraph{Base Case, $n = 1$.} 
    \texttt{OVPre} stores the input vector using $d$ bits and preprocessing time $d$. \texttt{OVOnl} solves the problem by checking if the input vector is orthogonal to the query vector in time $2d$.

    \paragraph{Base Case, $i=1$.} 
    \texttt{OVPre} creates only a bitmap of all $2^{d} = \binom{d}{\leq d}$ queries. For each query, it takes time $nd$ to compute the answer to the query for a total time of $nd2^d$. Lastly, \texttt{OVOnl} simply looks up the answer in the stored bitmap, which takes $d + 1 \leq 2d$ time and is always correct.

    \paragraph{Induction (Correctness).}
    Now, we consider the case where $n > 1$ and $i > 1$ by induction. First, let us argue for correctness.
    
    Given any query $q$ with $\norm{q}_0 < d/i$, \texttt{OVOnl} looks up the answer in the previously stored bitmap, so assume $\norm{q}_0 \geq d/i$.
    The query $q$ is orthogonal to a vector in $X$ if and only if $q$ is orthogonal to a vector in one of $X', X_1, \dots,$ or $X_k$. 
    To check if $q$ is orthogonal to a vector in $X'$, \texttt{OVOnl} first selects $d/i$ coordinates $C$ on which $q|_C = \vec{1}$ and considers the corresponding set of ``candidate vectors'' $Y_C = \{x \in X' : x|_C=\vec{0}\}$.
    Note that $q$ is orthogonal to a vector in $X'$ if and only if $q$ is orthogonal to a vector in $Y_C$, so \texttt{OVOnl} handles the vectors in $X'$ correctly.
    For each of the sets $X_j$ with $j \in [k]$, \texttt{OVOnl} checks to see if $q$ is orthogonal to the vectors in $X_j$ only on the coordinates $\overline{C_j}$ by applying itself recursively with $i' := i - 1$ and $d' := d-d/i$ and $n' = n^{1-1/i}$. This ensures that $i' = i - 1 \leq  d - d/i = d'$ since $i \leq d$. 
    By \cref{lem:part_alg_properties}, for every vector $x \in X_j$, $x|_{C_j} = \vec{0}$, so such an $x$ is orthogonal to $q$ if and only if $x|_{\overline{C_j}}$ is orthogonal to $q|_{\overline{C_j}}$. 
    Therefore, by the inductive hypothesis, the recursive call to \texttt{OVOnl} properly determines if $q|_{\overline{C_j}}$ is orthogonal to any vector in $X_j$. 
    
    \paragraph{Induction (Space).} Now, let us analyze space.
    \texttt{OVPre} writes down three pieces of information:
    \begin{enumerate}
        \item To handle sparse queries with $\norm{q}_0 < d/i$, \texttt{OVPre} writes a bitmap, one bit per query.
        There are $\binom{d}{< d/i}$ such queries. 
        \item For the pseudorandom part $X'$, by \cref{lem:part_alg_properties}, we know that for every $C \in \binom{[d]}{d/i}$, the corresponding $Y_C$ on line \ref{alg:ovpre_YC} has at most $|Y_C| \leq n^{1-1/i}$ vectors. Thus, for $X'$ \texttt{OVPre} writes down a list of at most $n^{1-1/i}$ many vectors for each of the $\binom{d}{d/i}$ choices of $C$, which requires at most $\binom{d}{d/i} \cdot n^{1-1/i} \cdot d$ bits.
        \item Lastly, for each structured part $\sigma_j$, \texttt{OVPre} writes down the corresponding $C_j$ and a recursively-constructed data structure for $\sigma_j$.
    Because $C_j \in \binom{[d]}{d/i}$, encoding $C_j$ requires $\log ( \binom{d}{d/i} ) \leq d$ bits.
    Applying our inductive hypothesis gives us that the recursively-constructed data structures each require at most $(i-1)\binom{d'}{\leq \frac{d'}{i-1}} \cdot (n')^{1 - \frac{1}{i-1}} d' = \binom{d(1-1/i)}{\leq \frac{d}{i}} \cdot n^{1 - 2/i} d\cdot (i-1)^2/i$ bits. 
    Therefore, the amount of space used for all of the recursive calls is at most
    \begin{align*}
        n^{1/i} \cdot \left(d+\binom{d(1-1/i)}{\leq \frac{d}{i}} \cdot n^{1 - 2/i} d\cdot \frac{(i-1)^2}{i}\right)
        & \leq dn^{1-1/i} + dn^{1-1/i} \binom{d(1-1/i)}{\leq \frac{d}{i}}\frac{(i-1)^2}{i} \\
        & =  dn^{1-1/i} \left(1 + \binom{d(1-1/i)}{\leq \frac{d}{i}}\frac{(i-1)^2}{i}\right)\;,
    \end{align*}
    where the inequality follows since $i \geq 2$.
    \end{enumerate}
    
    Therefore, in total, \texttt{OVPre} uses space at most
    \begin{align*}
        S &\leq \binom{d}{< d/i} + \binom{d}{d/i} \cdot n^{1-1/i} \cdot d +  dn^{1-1/i} \left(1 + \binom{d(1-1/i)}{\leq \frac{d}{i}} \frac{(i-1)^2}{i}\right) \\
        & \leq dn^{1-1/i} \cdot \binom{d}{\leq d/i} + dn^{1-1/i}   \left(1 + \binom{d(1-1/i)}{\leq \frac{d}{i}} \frac{(i-1)^2}{i}\right) \\
        & \leq idn^{1-1/i} \cdot \binom{d}{\leq d/i}\;.
    \end{align*}

    \paragraph{Induction (Preprocessing Time).} Now, let us analyze preprocessing time.
    \texttt{OVPre} computes three pieces of information:
    \begin{enumerate}
        \item It takes time at most $d+nd/i$ to compute the answer for each sparse query. 
        Therefore, across all sparse queries, the total amount of time taken is at most $(n+1)d\binom{d}{< d/i}$. 
        \item For the pseudorandom part $X'$, by \cref{lem:part_alg_properties}, we know that we can find  $X', X_1, \dots, X_k, C_1, \dots, C_k$ in time $\binom{d}{d/i} \cdot nd/i$. 
        Furthermore, it takes time at most $\binom{d}{d/i} \cdot nd/i$ to write down the list $Y_C$ of at most $n^{1-1/i}$ vectors for each of the $\binom{d}{d/i}$ choices of $C$.
        Thus, for the first two steps listed here, it takes time
        \begin{align*}
            \binom{d}{<d/i}\cdot (n+1) \cdot d + 2 \cdot \binom{d}{d/i} \cdot n \cdot d/i &= \binom{d}{\leq d/i} \cdot (n+1) \cdot d - \binom{d}{d/i} \cdot d \cdot \left(n+1-2n/i \right) \\
            &\leq \binom{d}{\leq d/i}\cdot (n+1) \cdot d.
        \end{align*}
        \item For each structured part $X_j$, \texttt{OVPre} creates a recursively-constructed data structure for $X_j$.
        Applying our inductive hypothesis gives us that each of the recursively-constructed data structures requires time at most
        \begin{align*}
            \binom{d'}{\leq \frac{d'}{i-1}}\cdot n' \cdot d' \cdot (i-1) &= \binom{d-d/i}{\leq d/i} \cdot n^{1-1/i} \cdot d \left(1 - \frac{1}{i} \right) \cdot (i-1) \\
            &\leq \binom{d}{\leq d/i} \cdot n^{1-1/i} \cdot d \cdot \left(i - \frac{3}{2} \right)
        \end{align*}
        because $i \geq 2$.
        Therefore, the amount of time used for all of the recursive calls is at most
        \begin{equation*}
            \binom{d}{\leq d/i} \cdot n \cdot d \cdot \left(i - \frac{3}{2} \right).
        \end{equation*}
    \end{enumerate}
    
    So in total, \texttt{OVPre} uses time at most
    \begin{align*}
        T_p &\leq \binom{d}{\leq d/i} \cdot (n+1) \cdot d + \binom{d}{\leq d/i} \cdot n \cdot d \left(i - \frac{3}{2} \right) \\
        &= \binom{d}{\leq d/i} \cdot d \left(n + 1 + in -\frac{3}{2}n \right).
    \end{align*}
    Because $n \geq 2$, $3n/2 \geq n + 1$, so we have
    \begin{equation*}
        T_p \leq \binom{d}{\leq d/i} \cdot d \cdot i \cdot n.
    \end{equation*}

    \paragraph{Induction (Query Time).} Lastly, we prove the bound on query time.
    First note that if the given query $q$ has $\norm{q}_0 < d/i$, then \texttt{OVOnl} simply looks up the answer in a bitmap; this takes time $d+1$.
    So assume that $\norm{q}_0 \geq d/i$.
    To determine if $q$ is orthogonal to any vector in the pseudorandom part, $\texttt{OVOnl}$ selects an arbitrary set $C$ on which $q|_C = \vec{1}$ and checks $q$ against the vectors in $Y_C$.
    Because $|Y_C| \leq n^{1-1/i}$, this takes at most time $(n^{1-1/i}+1)\cdot d \leq 2d \cdot n^{1-1/i}$.
    Then \texttt{OVOnl} recursively calls itself on at most $n^{1/i}$ many subproblems consisting of $n'$ vectors in $d'$ dimensions.
    By the inductive hypothesis, each takes time at most $2(i-1)d' \cdot (n')^{1-1/(i-1)} = 2\cdot n^{1-2/i} d \cdot (i-1)^2/i$.
    Thus, overall, \texttt{OVOnl} takes time at most
    \begin{align*}
        T &\leq \max(d + 1, 2d \cdot n^{1-1/i} + n^{1/i} \cdot 2\cdot n^{1-2/i} d \cdot (i-1)^2/i) \\
        &\leq 2d \cdot n^{1-1/i} + 2\cdot n^{1-1/i} d \cdot (i-1)^2/i \\
        &\leq 2id \cdot n^{1-1/i}.\qedhere
    \end{align*}
\end{proof}

\subsection{Applications}

In this section, we combine our algorithms along with known reductions to provide new data structures for the online versions of a variety of problems. Since our algorithms are deterministic, we are able to provide the same guarantees for all problems that have an efficient deterministic reduction to $\OnlineOV_{n, d}$. We collect such reductions in \cref{lem:everything_is_ov}. 

\begin{corollary}\label{apps:deterministic_applications}
   For all large enough $n, c$ and $\delta \geq \Omega(\log c/c)$, there are $(S, T)$-data structures with preprocessing time $T_p$ for each of the following problems: 
    \begin{align*} 
     & \PM_{n, c \log n}, \;\DNFEval_{n, c \log n}, \; \SubsetQuery_{n,  c \log n}, \\ & \ContainmentQuery_{n, c \log n}, \; (1+\eps)\text{-}\ORS_{n, O(\eps c \log n/\log u), u},
    \end{align*}
    where
    \begin{align*}
       T = n^{1-\Omega\left(\frac{1}{c\log c}\right)}   \quad \text{and} \quad S, \;T_p = \tilde{O}(n^{1+\delta}) \;.
    \end{align*}
\end{corollary}

\begin{proof}
    We combine \cref{lem:everything_is_ov} and \cref{cor:chan_comparison} to conclude the corollary.
\end{proof}

We can also obtain improved deterministic data structures for these problems in arbitrary dimensions using the same reduction. In~\cref{apps:deterministic_applications}, we only highlight the case when the dimension $d(n) = c \log n$ as it is where we obtain almost linear space and preprocessing time while having sublinear query time. 

In addition, since there is an efficient randomized reduction from $\ORS$ to $\OnlineOV$, we are able to obtain a randomized data structure for $\ORS$ in arbitrary dimensions. Note that the case where $d(n) \leq  c \log^2 n$ was already provided by \cite{C17}.  

\begin{corollary}\label{apps:ors}
    For large enough $n, d$ and $\eps < 1/2$, there is an $(S, T)$-data structure with preprocessing time $T_p$ for $\ORS_{n,d, u}$ where
    
    \begin{align*} T &\leq O(n^{1-\eps}d^2\log^2(u)), & S & \leq 2^{O(\eps d^2\log^2(u)\log(1/\eps))} \cdot n^{1-\eps} \cdot d^2\log^2(u), & \\ & \text{and}
    & T_p & \leq 2^{O(\eps d^2\log^2(u)\log(1/\eps))} \cdot n^{1+3\eps} \cdot d^2\log^2(u) \;.
    \end{align*}
\end{corollary}

Finally, we are also able to demonstrate a new offline algorithm for a version of $\OV$ where at least one set of input vectors from $A, B \subseteq \{0,1\}^{c\log n}$ is chosen at random from $\mathcal{B}_p$.  

\begin{corollary}\label{apps:avg-ov}
Let $A, B \subseteq \{0, 1\}^{c \log n}$ such that $A \sim \mathcal{B}_p$. Then there is an algorithm that solves $\OV$ on $A, B$ in expected time $n^{2-\Omega(1/\log c)}$.
\end{corollary}
\begin{proof}
    This follows from applying our data structure from \cref{cor:chan_comparison_avg} with inputs from $A$ and queries from $B$. This would take time at most 
    $n^{1+O(\log c / c)} + n \cdot n^{1-\Omega(1/\log c)} = n^{2-\Omega(1/\log c)}$.
\end{proof}

\section{Lower Bounds}\label{sec:hardness}
In this section, we use hardness conjectures against non-uniform algorithms for $k$-SAT (\cref{conj-nuseth}) and \HamPath{} (\cref{conj-nuhampath}) to prove polynomial lower bounds on the space and query complexities of data structures (with computationally unbounded preprocessing) for Sparse-\OnlineOV{} and \kSUM{}. By \cref{lem:everything_is_ov}, this also implies hardness of many other related problems.

\subsection{Space Lower Bounds from Non-Uniform SETH}\label{sec:nuseth_lb}
We begin by defining the hardest \OnlineOV{} instance: a family of sparse \OnlineOV{} instances for which efficiently checking orthogonality against the vectors in the instance would enable us to solve SAT for \emph{any} input formula. We then show that there exists an explicit hardest \OnlineOV{} instance, and we use it to prove space lower bounds for a range of data structure problems.

\begin{definition}[Hardest \OnlineOV{} Instance]\label{def:hardest_OnlinneOV}

For a function $w:=w(k)$ and a constant $\delta\in(0,1]$, a family $(\mathcal{D}_{N,w})_{N,w\in\Z^+}$ of \OnlineOV{} instances is a \emph{$(w,\delta)$-hardest \OnlineOV{} instance} if it satisfies the following properties.
\begin{itemize}
\item For every $N, w$, $\mathcal{D}_{N,w}$ is an $\OnlineOV{}$ instance whose input contains $N$ $w$-sparse vectors of dimension $O(N^{\delta})$.

\item Efficiently computable: there is a deterministic algorithm that, given $N$ and $w$, outputs $\mathcal{D}_{N,w}$ in time $\widetilde{O}(Nw)$.

\item Evaluating $(\mathcal{D}_{N,w})_{N,w\in\Z^+}$ solves $k$-SAT: There is a deterministic algorithm~$\mathcal{A}$ that, given a $k$-SAT formula~$\varphi$ on $n$ variables, outputs a vector $a_\varphi$ of dimension $N^{\delta}$ for $N=2^n$ and $w=w(k)$ such that $\varphi$ is satisfiable if and only if there is a vector in $\mathcal{D}_{N,w}$ that is orthogonal to $a_\varphi$. 
The running time of~$\mathcal{A}$ is $\widetilde{O}(N^\delta)$. %
\end{itemize}
\end{definition}

In the next lemma, we follow the high-level idea of~\cite{BGKMS23} to construct an explicit $(w,\delta)$-hardest \OnlineOV{} instance. 
\begin{lemma}\label{lem:poly_form_OV}
    For every constant $\delta \in (0,1]$, there is a $(w, \delta)$-hardest \OnlineOV{} instance with $w = \binom{k/\delta}{k}$. %
\end{lemma}

\begin{proof}
    We define the set $\mathcal{D}_{N,w}$ of $N$ input vectors, each of dimension $O(N^{\delta})$, forming our hardest \OnlineOV{} instance, solving all instances of $k$-SAT on $n$ variables.
   Let $\theta = k/\delta$, and let $V$ be the parts of an arbitrary partition of the $n$ variables of a $k$-CNF formula into $\theta$ many blocks, $V_1, \dots, V_\theta$, each of size $|V_i| = n/\theta$. 
   
   Each coordinate of our vector will represent an indicator from the following set of variables,

    \[
    X = \{x_{B, \mu_B} \colon B \in \binom{V}{k},\; \mu_B \in \bit^{nk/\theta}\} \;.
    \]
    Here, $B$ corresponds to a set of $k$ parts of our partition $V$, and $\mu_B$ corresponds to an assignment to the variables in all the parts in $B$.
    The total number of variables in $X$ is 
    \[v = \binom{\theta}{k}\cdot 2^{nk/\theta}
    = O(2^{\delta n}) = \Ot(N^{\delta}) \;.\]
    
    We now let the $i$th coordinate of our vectors represent the $i$-th lexicographic variable in $X$, and define our $N$, vectors as follows.

    For every assignment $\mu \in \bit^n$, we define a vector $X_{\mu} \in \mathcal{D}_{N, w}$. The coordinates of $X_\mu$ are labeled by pairs $(B, \mu_B)$. For $X_\mu$, we set the coordinate corresponding to $(B, \mu_B)$ to one if and only if $\mu_B$ is the assignment $\mu$ restricted to the variables in $B$. We observe that each $X_\mu$ has exactly $\binom{V}{k}  = \binom{\theta}{k} = O(1)$ non-zero entries.

    Note, importantly, that $\mathcal{D}_{N, w}$ does \emph{not} depend on a particular $k$-CNF, and it only depends on $n$ and $k$.

Now, given a $k$-CNF formula $\phi = C_1 \wedge \cdots \wedge C_m$ on $n$ variables, we determine whether $\varphi$ is satisfiable or not by checking orthogonality against the vectors in $\mathcal{D}_{N,w}$. 

We start by defining a bipartite graph $G = (V \sqcup C, E)$ as follows.
The vertices on the left, $V$, are the parts of our partition of the variables of $\phi$ into $\theta$ many blocks. The vertices on the right, $C$, are the clauses of $\phi$. We add an edge $\{V_i, C_j\}$ if $V_i$ contains a variable (or its negation) from $C_j$.
And finally, we add arbitrary edges to $G$ such that the degree of every $C_j$ is exactly $k$.

We now define the query vector $q_{\varphi}$, as follows: set the coordinate corresponding to $x_{B, \mu_B} = 0$, if the clauses in $\varphi$ whose neighborhood in $G$ is $B$ are all satisfied by the assignment $\mu_B$; otherwise, set $x_{B,\mu_B} = 1$.

Note that $q_{\varphi}$ is orthogonal to a vector $X_{\mu}$ in $\mathcal{D}_{N, w}$ if and only if there is at least one assignment $\mu$ for which all coordinates corresponding to $x_{B, \mu_B} = 0$ in $q_{\varphi}$. Such a $\mu$ is a satisfying assignment for $\varphi$. Note that for a given $k$-CNF formula~$\phi$, setting $x_{B, \mu_B}$ to the correct values can be done in time nearly linear in the number of variables $\widetilde{O}(v)$.

In conclusion, for every $N$ and $w$, we give a deterministic construction that in time $\widetilde{O}(N)$ outputs an instance $\mathcal{D}_{N,w}$ of $w$-Sparse-$\OnlineOV_{N, d}$  with sparsity $w=\binom{k/\delta}{k}$, and dimension $d = O(N^\delta)$. We also provide an efficient deterministic algorithm that, given any $k$-CNF $\varphi$, outputs a query vector $q_{\varphi}$ whose orthogonality to $\mathcal{D}_{N,w}$ characterizes the satisfiability of $\varphi$.
\end{proof}

We now use \cref{lem:poly_form_OV} to show that
under Non-Uniform SETH, data structures with computationally unbounded preprocessing even for the \emph{batch} version of Sparse \OnlineOV{} require either super-polynomial space or nearly linear time per query.

\begin{theorem}\label{thm:batchDNF}
Under \rm{NUSETH}, for all constants $c,\alpha>0$ and $\delta,\gamma\in(0,1)$, there exists constant~$w$ such that no batch data structure with computationally unbounded preprocessing can answer a set of $n^\alpha$ queries of $w$-Sparse-$\OnlineOV_{n,d}$ for $d=n^\delta$ in space $n^c$ and time $n^{1-\gamma}$ per query.
\end{theorem}

\begin{proof}
     Assume that there exist constant $c, \delta, \alpha, \gamma$ and a data structure that solves $n^\alpha$ queries of $w$-Sparse-$\OnlineOV_{n,d}$ for $d=n^\delta$ in space $n^c$ and time $n^{1-\gamma}$ per query. Let 
    \begin{align*}
    \beta &= \min(1/(2c),\; 1/(\alpha+1))\;,\\
    \eps &= \min(1/2, \; \beta(1-\delta), \; \beta\gamma) \;.
    \end{align*}
    We will construct a non-uniform algorithm solving $k$-SAT on $m$ variables in time $2^{(1-\eps)m}$ for every constant~$k$, which refutes NUSETH.

    Let $\phi$ be a $k$-CNF formula on $m$ variables.  We apply the standard branching argument based on the downward self-reducibility of $k$-SAT. Namely, for every assignment $\cc{u} \in \bit^{(1-\beta)m}$ to the first $(1-\beta)m$ variables of $\phi$, let $\phi_{\cc{u}}$ be the remaining $k$-CNF formula on $\beta m$ variables. Note that $\phi$ is satisfiable if and only if at least one $\phi_{\cc{u}}$ is satisfiable. 

    By \cref{lem:poly_form_OV}, we obtain an $\OnlineOV_{N,d}$ instance $\mathcal{D}_{N,w}$ with $N=2^{\beta m}$ vectors in $d=N^\delta$ dimensions to solve $k$-CNF formulas on $\beta m$ variables.
    For each $\cc{u} \in \bit^{(1-\beta)m}$, we wish to determine if $\phi_{\cc{u}}$ is satisfiable.  By \cref{def:hardest_OnlinneOV}, we have that $\phi_{\cc{u}}$ is satisfiable if and only if there exists $v \in \mathcal{D}_{N, w}$ with $v \perp a_{\cc{u}}$. Let $\mathcal{Q}$ be the set of the vectors $a_{\cc{u}}$ for all formulas $\phi_{\cc{u}}$.
    Note that size of $|\mathcal{Q}| = 2^{(1-\beta)m}$, and that by \cref{def:hardest_OnlinneOV} it can be computed in time $\widetilde{O}(|\mathcal{Q}|\cdot N^\delta)=2^{(1-\beta+\beta\delta)m} = 2^{(1-\beta(1-\delta))m}$.

    The constructed non-uniform algorithm for $k$-SAT uses the advice string $\sigma$ and the query algorithm for the $w$-Sparse-$\OnlineOV_{n,d}$ problem to determine if there is an orthogonal pair of vectors in $\mathcal{Q}$ and $D$.

    First we argue the correctness of the designed non-uniform algorithm for $k$-SAT. First, by self-reduction of $k$-SAT, $\phi$ is satisfiable if and only if one of $\phi_{\cc{u}}$ is satisfiable. By \cref{def:hardest_OnlinneOV}, $\phi_{\cc{u}}$ is satisfiable if and only if there exists $v \in \mathcal{D}_{N, w}$ with $v \perp a_{\cc{u}}$.

    Now we analyze the running time of the designed algorithm for $k$-SAT. The length of the advice string $\sigma$ is $S=N^c = 2^{\beta c m}\leq 2^{m/2}\leq 2^{(1-\eps)m}$ by the choice of $\beta\leq 1/(2c)$ and $\eps\leq 1/2$. The time needed to construct all queries $\mathcal{Q}$ is $2^{(1-\beta(1-\delta))m}\leq 2^{(1-\eps)m}$ by the choice of $\eps\leq \beta(1-\delta)$. Finally, since the number of queries $|K|=2^{(1-\beta)m}\geq 2^{(\alpha \beta)m}=N^{\alpha}$ is no less than the number of queries in the assumed batch data structure, the amortized running time per query is $N^{1-\gamma}$, and the total running time of the \OnlineOV{} query algorithm for all queries is $|\mathcal{Q}|\cdot N^{1-\gamma}=2^{(1-\beta\gamma)m}\leq 2^{(1-\eps)m}$ by $\eps\leq \beta\gamma$.
\end{proof}

Finally, we observe that since some problems admit sparsity-preserving reductions from OnlineOV, we can establish hardness results for their sparse variants in the sub-linear dimension regime from \cref{thm:batchDNF}.

\begin{corollary}\label{cor:sparse_ds}
Under \rm{NUSETH}, for all constant $c,\alpha>0$ and $\delta, \gamma \in(0,1)$, there exists constant $w$ such that no $(n^c, n^{1-\gamma})$-data structure with computationally unbounded preprocessing for any of the following problems with $d=n^\delta$. 
\begin{align*}&  w\text{-Sparse-}\SubsetQuery_{n, d}, \; w\text{-Sparse-}\ContainmentQuery_{n, d}, \text{Monotone-}w\text{-}\DNFEval_{n,d} \; .
    \end{align*}
\end{corollary}

\begin{remark}\label{rem:ov_low_dim}
    We remark that, if one is willing to relax the sparsity requirement, then the lower bounds can be established in significantly smaller dimension. In particular, \cref{lem:poly_form_OV} holds even for $\delta = o(1)$. For example, setting $\delta = k/N$, the resulting hardest \OnlineOV{} instance has dimension $\poly(\log (N))$, albeit at the cost of the vectors becoming dense. Moreover, this instance coincides exactly with the hard \OnlineOV{} instance used in \cite{AV21} to prove lower bounds under NUSETH. Thus, \cref{lem:poly_form_OV} can be viewed as a generalization of \cite{AV21}, establishing hardness of \OnlineOV{} even for sparse instances and revealing a smooth tradeoff between sparsity and dimension in the construction of the hardest \OnlineOV{} instance. 
\end{remark}

We state the following corollary that follows from \cite{AV21}, but can be recovered from \cref{thm:batchDNF} by constructing a $(w, \delta)$-hardest \OnlineOV{} instance with $\delta = k/N$.
 
\begin{corollary}[\cite{AV21}]\label{cor:av_21}
    Under \rm{NUSETH}, for all constants $c,\alpha>0$, $\gamma\in(0,1)$, there exists a constant $\rho>0$ such that no batch data structure with computationally unbounded preprocessing can answer a set of $n^\alpha$ queries of $\OnlineOV_{n,d}$ for $d=O(\log(n)^\rho)$ in space $n^c$ and time $n^{1-\gamma}$ per query.
\end{corollary}

\subsection{Space Lower Bounds from Non-Uniform HamPath}
In this section, we show that an efficient data structure for \kSUM{} with preprocessing would imply a non-uniform algorithm for \HamPath{} with running time exponentially smaller than~$2^n$. First, we define the \kSUM{} problem with preprocessing.

\begin{definition}\label{def:3sum_pre}
    \kSUM{} with preprocessing is a problem to be solved by a pair of deterministic algorithms $(\mathcal{P}, \mathcal{A})$ in two phases.

    \begin{description}
        \item{\bf Preprocessing phase.} In the first phase, the algorithm $\P$ receives $k$ sets: $\mathcal{X}_1, \hdots , \mathcal{X}_k$ each consisting of $N$ integers, and preprocesses them into a data structure $\sigma\in\{0,1\}^S$ consisting of~$S$ bits.
        
        \item{\bf Query phase.} In the second phase, the algorithm $\A$ receives a query $\mathcal{X}_1' \subseteq \mathcal{X}_1, \hdots, \mathcal{X}_k' \subseteq \mathcal{X}_k$, (query access to) the data structure~$\sigma$, and runs in time $T$. $\A$ outputs $1$ if there exist $x_1 \in \mathcal{X}_1', \hdots, x_k \in \mathcal{X}_k'$ such that $x_1 + \hdots + x_{k-1} = x_k$, and it outputs~$0$ otherwise.
    \end{description}
    Such a data structure is called an $(S, T)$-data structure for \kSUM{} with preprocessing.
\end{definition}

\begin{theorem}\label{thm:hampath}
Under the \rm{Non-Uniform HamPath Conjecture}, for every constant $k\geq3$ and ${\delta < \frac{1}{H(1/k)}}$ there is no $(N^{\delta}, N^{\delta})$-data structure with computationally unbounded preprocessing solving  \kSUM{}.

\end{theorem}

\begin{proof}
Let $G$ be an instance of \HamPath{} on~$n$ vertices.
    If $n$ is not a multiple of~$k$, assuming $r=n \bmod k$, we add to the graph a path on $(k-r)$ vertices, and add an edge from each vertex of the original graph to the first vertex of the path. The new graph has a Hamiltonian path if and only if the original graph does. Thus, in the following we will assume that $n$ is a multiple of~$k$.
    
    Assume, towards a contradiction, that there is an $(N^{\delta}, N^{\delta})$-data structure $(\mathcal{P}, \mathcal{A})$ for \kSUM{} with preprocessing for $1\leq \delta < \frac{1}{H(1/k)}$. Then on input
     sets $\mathcal{X}_1, \ldots \mathcal{X}_k$ of size $N$, the computationally unbounded $\mathcal{P}$ outputs a data structure $\sigma\in\{0,1\}^{N^{\delta}}$, and for every query $q = \{\mathcal{X'}_1 \subseteq \mathcal{X}_1, \ldots, \mathcal{X'}_k  \subseteq \mathcal{X}_k\}$, $\mathcal{A}^{\sigma}(q)$ runs in time $T = N^{\delta}$ and outputs $1$ if and only if there exist $x'_1 \in \mathcal{X'}_1, \ldots, x'_k \in \mathcal{X'}_k$ such that $x'_1 + \hdots + x'_{k-1} = x'_{k}$.

   We describe a (non-uniform) reduction from \HamPath{} to \kSUM{} with preprocessing. For a set $S\subseteq[n]$, let $\Int(S)$ denote the integer from $\{0,\ldots,2^n-1\}$ whose binary representation is the indicator vector of~$S$. 
   Let $\mathcal{S}_{n/k}$ be the collection of all subsets of $[n]$ of size $n/k$. We run the preprocessing algorithm $\P$ for \kSUM{} on the following sets of size $N = |\mathcal{S}_{n/k}|$.
    \begin{align*}
        \mathcal{X}_1, \, \ldots, \, \mathcal{X}_{k-1} &= \{ \Int(S) \colon S \in \mathcal{S}_{n/k}  \}\;,\\
        \mathcal{X}_k &= \{ 2^{n} -1 - \Int(S) \colon S \in \mathcal{S}_{n/k}  \}\;.
    \end{align*}
    We store the returned data structure $\sigma \in \{0,1\}^{N^\delta}$ as the non-uniform advice of our algorithm for \HamPath{}.
    (We note that the sets $\mathcal{X}_i$ only depend on~$n$ and do \emph{not} depend on a \HamPath{} instance.)

Next, we describe a RAM algorithm solving \HamPath{} for a given $n$-vertex graph using the advice string~$\sigma$ and the query algorithm $\A$ for \kSUM{} with preprocessing. First, we use the dynamic programming algorithm of \cref{prop:hk} to  find all simple paths of length exactly $n/k$ in time $\binom{n}{\leq n/k}\cdot\poly(n)\leq 2^{nH(1/k)}\cdot\poly(n)$.
For a set $S\in\mathcal{S}_{n/k}$, and vertices $u,v\in[n]$, we define the following variable,
    \[
    D[S, u, v] = 
\begin{cases}
    1 & \text{if $u\in S$ and $\exists$ a simple path from $u$ to a neighbor of~$v$ using all vertices in $S$,} \\
    0  & \text{otherwise.}
\end{cases}
\]
Using the fact that each Hamiltonian path can be partitioned into $k$ disjoint simple paths of length $n/k$, we have the following. $G$ has a Hamiltonian path if and only if there exists a set of $k+1$ vertices $v_1, \ldots, v_{k+1}\in[n]$ and a collection of $k$ disjoint sets:  $S_1, \ldots, S_k \in \mathcal{S}_{n/k}$ such that
    \[\bigwedge_{i \in [k]}D[S_i, v_i, v_{i+1}] = 1\;.\]

    Next, we design $n^{k+1}$ \kSUM{} queries that will determine whether $G$ has a Hamiltonian path. For every $q = (v_1, \hdots, v_{k+1})$, we compute the following subsets:
    \begin{align*}
        \mathcal{X}^{q}_i &= \{ \Int(S) \colon S \in \mathcal{S}_{n/k}, \; v_i \in S, \text{ and }   D[S, v_i, v_{i+1}]=1 \} \text{ for } i \in [1, k-1]\;,\\
        \mathcal{X}^{q}_{k} &= \{2^{n} - 1 - \Int(S) \colon S \in \mathcal{S}_{n/k}, \; v_{k} \in S,  \text{ and }   D[S, v_{k}, v_{k+1}]=1\}\;.
    \end{align*}

    Since for every $i\in[k]$, $\mathcal{X}^{q}_{k}\subseteq \mathcal{X}_{k}$, we can use the query algorithm $\A$ for the \kSUM{} with preprocessing to check if there exist $x_1 \in \mathcal{X}^q_1, \hdots, x_k \in \mathcal{X}^q_k$ such that $x_1 + \hdots + x_{k-1} = x_k$. Or, equivalently, if there exist $v_1,\ldots, v_{k+1}\in[n]$ and $S_1,\ldots,S_k\in\mathcal{S}_{n/k}$ such that 
    \begin{align}\label{eq:hampath}
    \Int(S_1)+\ldots+\Int(S_k)=2^n-1 \text{\;\;\;and\;\;\;}
    D[S_i, v_i, v_{i+1}]=1 \text{ for all } i\in[k].
    \end{align}
    The algorithm reports that there exists a Hamiltonian Path in $G$ if and only if at least one of the \kSUM{} queries returns~$1$,
    \[\bigvee_{q = (v_1, \hdots, v_{k+1)} \in [n]^{k+1}}\A^{\sigma}({\mathcal{X}^{q}_1, \ldots, \mathcal{X}^{q}_{k}})\;. \]
    
    This completes the reduction. We first argue the correctness of the reduction. Namely, we show that \cref{eq:hampath} holds for one of the queries if and only if $G$ has a Hamiltonian Path.
    Let $(a_1,\ldots,a_n)$ be a Hamiltonian path in~$G$. For $i\in[k]$, we set $v_i=a_{1+(i-1)n/k}$, and we set $v_{k+1}$ to be an arbitrary neighbor of $a_n$. Similarly, $S_i = \{a_{1+(i-1)n/k}, \ldots,  a_{in/k}\}$ for every $i\in[k]$. Note that $D[S_i, v_i, v_{i+1}]=1$ for all $i\in[k]$. Finally, since all $S_i$ are disjoint, we have that $\Int(S_1)+\ldots+\Int(S_k)=2^n-1$, which implies that \cref{eq:hampath} holds for the query $q=(v_1,\ldots,v_{k+1})$. 
    
    For the other direction, let us assume that \cref{eq:hampath} holds for some $q = (v_1, \hdots, v_{k+1})$ and $S_1,\ldots, S_{k}$. We use the fact that the sum of $k$ numbers, each having exactly $n/k$ ones in their binary representation, has $n$ ones in the binary representation if and only if the supports of the ones are disjoint. Thus, \cref{eq:hampath} implies that all $S_i$ are disjoint. This, together with $D[S_i, v_i,v_{i+1}]=1$ for all $i\in[k]$, gives us a Hamiltonian path in~$G$.

    We now analyze the complexity of the presented algorithm for \HamPath{}. The non-uniform advice has length $N^\delta=2^{\delta nH(1/k)}\cdot\poly(n)$. The running time of the algorithm is bounded by the running of the dynamic programming algorithm of \cref{prop:hk} and $n^{k+1}$ queries to $\A$, 
    \[
    2^{n H(1/k)}\cdot\poly(n) + n^{k+1}\cdot N^\delta = (2^{nH(1/k)}+2^{\delta nH(1/k)})\cdot\poly(n) 
    =(2^{\delta nH(1/k)})\cdot\poly(n) 
    \;.
    \]
    
    If $\delta < 1/H(1/k)$, by setting $\gamma = 1- \delta \cdot H(1/k)>0$, we obtain a $2^{\delta H(1/k) n} \cdot\poly(n)  = 2^{(1-\gamma)n} \cdot\poly(n)$ non-uniform algorithm for \HamPath{}, and refute the Non-uniform Hamiltonian Path Conjecture.
\end{proof}

\begin{corollary}
    Under the Non-Uniform HamPath Conjecture, there is no $(N^{1.088}, N^{1.088})$-data structure with computationally unbounded preprocessing solving  \ThreeSUM{}.
\end{corollary}

We remark that \cref{thm:hampath} also gives a \emph{uniform} reduction from \HamPath{} to the \emph{offline} version of \kSUM{}.

\begin{corollary}
    For every $k\geq3$ and every ${\delta < \frac{1}{H(1/k)}}$ there is a $\gamma>0$, such that an algorithm solving \kSUM{} on lists of size $N$ in time $N^{\delta}$ implies an algorithm solving \HamPath{} in time $2^{(1-\gamma)n}$.
\end{corollary}

\section*{Acknowledgments}
We are very grateful to Virginia Vassilevska Williams and an anonymous SODA reviewer who pointed out a simpler proof~\cite{AV21} of the conditional lower bound for \OnlineOV{} in the dense regime (\cref{cor:av_21}). We also thank all SODA reviewers for their helpful comments.

This research is supported by the National Science Foundation CAREER award (grant CCF-2338730).
\bibliographystyle{alpha}
\bibliography{refs}

\end{document}